\newif\ifIEEE
\newif\ifPAGELIMIT
\IEEEtrue
\PAGELIMITfalse
\ifPAGELIMIT
\IEEEtrue
\fi
\ifIEEE
    \ifPAGELIMIT
        \documentclass[conference,final,letterpaper]{IEEEtran}
    \else
        \documentclass[journal,final,letterpaper]{IEEEtran}
    \fi
    \newcommand{\bibauthor}[1]{#1}
    \newcommand{\bibpaper}[1]{``#1''}
    \newcommand{\Footnotetext}[2]
    {
        \begin{figure}[!b]
        \footnotesize\vspace{-3ex}\hrulefill\hfill
        \makebox[0em]{}\hfill\makebox[0em]{}%
                                          \par${}^{#1}$ #2\vspace{-.6ex}
        \end{figure}
        \addtocounter{figure}{0}
     }
\else
    \documentclass[10pt]{article}
    \newcommand{\bibauthor}[1]{\textsc{#1}}
    \newcommand{\bibpaper}[1]{\textsl{#1}}
    \newenvironment{IEEEkeywords}{\begin{small}%
                                  \textbf{Index Terms} ---}{\end{small}}
\fi
\usepackage{amsfonts,bm}
\usepackage{amsthm,amsmath,amssymb}
\usepackage{pict2e}
\usepackage{color}
\newcommand{\bibbook}[1]{\textit{#1}}

\newcommand{\bibperiodical}[1]{\textit{#1}}

\ifPAGELIMIT
\fi

\newtheorem{theorem}{\indent Theorem}

\newtheorem{lemma}[theorem]{\indent Lemma}
\newtheorem{corollary}[theorem]{\indent Corollary}

\theoremstyle{remark}
\newtheorem{remark}{\indent Remark}
\theoremstyle{definition}
\newtheorem{example}{\indent Example}
\renewcommand{\mathbf}[1]{{\bm{#1}}}     
\newcommand{\bldzero}{{\mathbf{0}}}
\newcommand{\bldone}{{\mathbf{1}}}
\newcommand{\blda}{{\mathbf{a}}}
\newcommand{\bldb}{{\mathbf{b}}}
\newcommand{\bldc}{{\mathbf{c}}}
\newcommand{\blde}{{\mathbf{e}}}
\newcommand{\bldh}{{\mathbf{h}}}
\newcommand{\bldg}{{\mathbf{g}}}
\newcommand{\bldm}{{\mathbf{m}}}
\newcommand{\bldr}{{\mathbf{r}}}
\newcommand{\blds}{{\mathbf{s}}}
\newcommand{\bldu}{{\mathbf{u}}}
\newcommand{\bldx}{{\mathbf{x}}}
\newcommand{\bldy}{{\mathbf{y}}}

\newcommand{\bldalpha}{{\mathbf{\alpha}}}
\newcommand{\bldrho}{{\mathbf{\rho}}}
\newcommand{\randomH}{{\mathbf{\mathit{H}}}}
\newcommand{\randomP}{{\mathbf{\mathit{P}}}}
\newcommand{\code}{{\mathcal{C}}}
\newcommand{\Ball}{{\mathcal{B}}}
\newcommand{\rank}{{\mathrm{rank}}}
\newcommand{\GF}{{\mathrm{GF}}}
\newcommand{\GRS}{{\mathrm{GRS}}}
\newcommand{\MDS}{{\mathrm{MDS}}}
\newcommand{\light}{{\mathrm{light}}}
\newcommand{\Event}{{\mathcal{A}}}
\newcommand{\DD}{{\mathcal{D}}}
\newcommand{\LL}{{\mathcal{L}}}

\newcommand{\Support}{{\mathsf{Supp}}}
\newcommand{\Prob}{{\mathsf{Prob}}}
\newcommand{\PP}{{\mathsf{P}}}
\newcommand{\Integers}{{\mathbb{Z}}}

\newcommand{\Realfield}{{\mathbb{R}}}

\newcommand{\weight}{{\mathsf{w}}}
\newcommand{\entropy}{{\mathsf{h}}}
\newcommand{\Entropy}{{\mathsf{H}}}

\newcommand{\bigcupdot}%
                 {{\textstyle{\bigcup\!\!\!\!\!\hspace{0.25ex}\cdot\;}}}

\newcommand{\Title}{Higher-Order MDS Codes}
\newcommand{\Namea}{Ron M. Roth}
\newcommand{\Addressa}{Computer Science Department}
\newcommand{\Addressatwo}{Technion, Haifa 3200003, Israel}

\newcommand{\Emaila}{ronny@cs.technion.ac.il}
\newcommand{\Grant}{This work was supported by 
                    Grant~1713/20 from
                    the Israel Science Foundation.}
\newcommand{\Thnxa}{\Namea\ is with the \Addressa, \Addressatwo.
                    Email: \Emaila}

\begin{document}
\ifIEEE
    \title{\Title}
       \ifPAGELIMIT
\else
           \author{\Namea
                   \thanks{\Grant}
                   \thanks{\Thnxa}
           }
       \fi
\else
    \title{\textbf{\Title}\thanks{\Grant}}
    \author{\textsc{\Namea}\thanks{\Thnxa}}
    \date{}
\fi
\maketitle

\begin{abstract}
An improved Singleton-type upper bound is presented for
the list decoding radius of linear codes, in terms of
the code parameters $[n,k,d]$ and the list size~$L$.
$L$-MDS codes are then defined as codes that attain this bound
(under a slightly stronger notion of list decodability),
with $1$-MDS codes corresponding to ordinary linear MDS codes.
Several properties of such codes are
\ifPAGELIMIT
\else
presented; in particular, it is shown
\fi
that the $2$-MDS property is preserved under duality.
Finally, explicit constructions for $2$-MDS codes are presented
through generalized Reed--Solomon (GRS) codes.
\end{abstract}

\ifPAGELIMIT
\else
\fi

\section{Introduction}
\label{sec:introduction}

Hereafter, we let~$F$ be the finite field $\GF(q)$.
Let~$\code$ be a code in $F^n$ and let $L \in \Integers^+$
and $\tau \in \Integers_{\ge 0}$ be given.
We say that~$\code$ is \emph{$(\tau,L)$-list decodable} if
for every $\bldy \in F^n$ there are no $L+1$ distinct codewords
of~$\code$ at Hamming distance${} \le \tau$ from~$\bldy$, i.e.,
\[
\bigl| \code \cap \left( \bldy + \Ball(n,\tau) \right) \bigr| \le L ,
\]
where $\Ball(n,\tau)$ is the set of vectors in $F^n$
with Hamming weight${} \le \tau$.
When~$\code$ is a linear $[n,k]$ code over~$F$,
it is $(\tau,L)$-list decodable if
there are no $L+1$ distinct vectors
\ifPAGELIMIT
\else
\[
\blde_0, \blde_1, \ldots, \blde_L \in \Ball(n,\tau)
\]
that are in the same coset of~$\code$ within $F^n$.
\fi
Equivalently, if~$H$ denotes a parity-check matrix of~$\code$,
the syndromes $H \blde_m^\top$ for $m \in \{ 0, 1, \ldots, L \}$
are not all equal.

The notion of list decoding was introduced by Elias~\cite{Elias1}
and has been an active area of research
in the last 25 years, where the focus has been primarily 
on constructing list decodable codes with efficient
decoding algorithms; among the most notable contributions
one can mention~\cite{GR}, \cite{GS}, \cite{PV},
and~\cite{Sudan}.
In addition, several papers presented bounds on
the parameters of list decodable codes, mostly
in an asymptotic setting~\cite{Blinovskii},
\cite{Guruswami}, \cite{GN}, \cite{GV}.

This work will be focusing on non-asymptotic bounds for list
decoding and on characterizations of (linear) codes that attain them.
\ifPAGELIMIT
\else
We recall next the well-known sphere-packing bound
for list decoding, which was proved in~\cite{Elias2}.
We use $V_q(n,\tau)$ to denote the size (volume) of $\Ball(n,\tau)$:
\[
V_q(n,\tau) = |\Ball(n,\tau)|
= \sum_{i=0}^\tau \binom{n}{i} (q-1)^i .
\]

\begin{theorem}[List decoding sphere-packing bound]
\label{thm:SpherePacking}
If~$\code \subseteq F^n$ is $(\tau,L)$-list decodable, then
\[
|\code| \le L \cdot \frac{q^n}{V_q(n,\tau)} .
\]
\end{theorem}

In particular, if~$\code$ is a linear $[n,k]$ code over~$F$, we have:
\begin{equation}
\label{eq:SpherePacking}
q^{n-k} \ge \frac{1}{L} \cdot V_q(n,\tau)
\ge
\frac{1}{L} \cdot \left( \frac{n(q-1)}{\tau} \right)^\tau .
\end{equation}
\fi

The following theorem was proved in~\cite[Theorem~1.2]{ST}
(see also~\cite[Theorem~2.6]{RV}).

\begin{theorem}[List decoding Singleton bound]
\label{thm:Singleton}
If~$\code \subseteq F^n$ is $(\tau,L)$-list decodable, then
\begin{equation}
\label{eq:Singleton1}
|\code| \le L \cdot q^{n - \tau - \lfloor \tau/L \rfloor} .
\end{equation}
In particular, if~$\code$ is a linear $[n,k]$ code over~$F$
and $L < q$, then
\[
\tau + \left\lfloor \frac{\tau}{L} \right\rfloor
\le n-k .
\]
\end{theorem}

   From
$\left\lfloor\tau/L\right\rfloor = \left\lceil(\tau-L+1)/L\right\rceil$
we get that the latter inequality is equivalent to
\begin{equation}
\label{eq:Singleton2}
\tau \le \frac{L(n-k) + L-1}{L+1} .
\end{equation}
As part of this work, we present
conditions under which~(\ref{eq:Singleton2}) can be improved to
\begin{equation}
\label{eq:Singleton3}
\tau \le \frac{L(n-k)}{L+1} ,
\end{equation}
and obtain constructions that attain
either~(\ref{eq:Singleton2}) or~(\ref{eq:Singleton3}).

\ifPAGELIMIT
\else
\begin{remark}
\label{rem:comparison}
Looking at the floor values of
the right-hand sides of~(\ref{eq:Singleton2}) 
and~(\ref{eq:Singleton3}), it is fairly easy to see that the latter
is smaller by~$1$ than the former, except when
$n-k \equiv 0 \; \textrm{or} \; L \pmod{(L+1)}$.\qed
\end{remark}

\begin{remark}
\label{rem:SpherePacking}
Theorem~\ref{thm:SpherePacking}
is sometimes stronger than~(\ref{eq:Singleton3}).
E.g., from~(\ref{eq:SpherePacking}) it follows
that the bound~(\ref{eq:Singleton3})
can be attained for a given~$L$ only if
\[
q = \Omega \left( \left( \frac{n}{\tau} \right)^L \right)
= \Omega \left( \left( \frac{n}{n-k} \right)^L \right)
\]
(where the hidden constants in the $\Omega(\cdot)$ terms
depend on~$L$).\qed
\end{remark}
\fi

Improving the bounds~(\ref{eq:Singleton1}) and~(\ref{eq:Singleton2})
is also the subject of the very recent work~\cite{GST}.
In particular, it is shown there that
when $L \le q$ and $L^2 \le \tau < L n/(L+1)$,
the multiplier~$L$ in~(\ref{eq:Singleton1}) can be replaced by
$1 + O(L/\tau)$ (or even by~$1$ when $\tau \equiv L-1 \pmod{L}$).
In Remark~\ref{rem:gst} below,
we will say more about the results of~\cite{GST}
and their relationship with our work.

The recent paper~\cite{ST} shows that
there exist generalized Reed-Solomon (GRS) codes
that attain the bound~(\ref{eq:Singleton3}) for $L = 2, 3$,
provided that the field size is sufficiently large (relative to~$n$).
In addition,
that paper presents an explicit construction of GRS codes
that attains~(\ref{eq:Singleton3}) for $L = 2$
(although the size of the underlying field grows like $2^{k^n}$).
Some of our results herein build upon~\cite{ST}.
\ifPAGELIMIT
\else
We also mention the paper~\cite{JH}
where the authors make a heuristic argument that could suggest that
when~$\tau$, $L$, $n$, and~$k$ satisfy the strict inequality
\[
\tau + \frac{\tau}{L} < n-k ,
\]
then, for sufficiently large~$q$,
there exist linear $[n,k]$ codes over~$F$ that are
$(\tau,L)$-list decodable.

One simple construction that attains the bound~(\ref{eq:Singleton2}) is
the repetition code of certain lengths.

\begin{example}
\label{ex:repetition}
For any given $L, u \in \Integers^+$,
the bound~(\ref{eq:Singleton2}) is attained
by the $[n,1]$ repetition code~$\code$ of length $n = (L + 1)u - 1$,
which is $(\tau,L)$-list decodable for
$\tau = (L n - 1)/(L+1) = L u - 1$.
Specifically, for any $\bldy \in F^n$, the intersection
$\code \cap \left( \bldy + \Ball(n,\tau) \right)$
contains all the codewords
$(c \, c\, \ldots \, c) \in \code$ that agree with~$\bldy$
on at least $n-\tau =  u$ coordinates; clearly, there can be
no more than $\lfloor n/u \rfloor = L$ such codewords
(see also~\cite[p.~175]{RR}).\qed
\end{example}
\fi

\ifPAGELIMIT
\else
In Sections~\ref{sec:intro-bounds}
and~\ref{sec:intro-LMDS} we describe the results of this work.

\subsection{Summary of improved bounds}
\label{sec:intro-bounds}

We first prove in Section~\ref{sec:bounds} the next theorem,
which
\fi
states that when $L < q$,
the bound~(\ref{eq:Singleton3}) can be attained
(or surpassed) only by MDS codes.
Hereafter, $[a:b]$ stands
for the
\ifPAGELIMIT
\else
integer subset
\fi
$\left\{ i \in \Integers \,:\, a \le i \le b \right\}$,
with $[b]$ being a shorthand notation for $[1:b]$.

\begin{theorem}
\label{thm:onlyMDS}
Given $L \in [q-1]$, let~$\code$ be a linear $[n,k]$ code over~$F$
and let $\tau \in \Integers_{\ge 0}$ be such that
\begin{equation}
\ifPAGELIMIT
\else
\label{eq:onlyMDS}
\fi
\tau \ge \frac{L(n-k)}{L+1} .
\end{equation}
Then~$\code$ is $(\tau,L)$-list decodable
\underline{only if} $\code$ is MDS.
\end{theorem}

\ifPAGELIMIT
\else
\begin{remark}
\label{rem:tau<d}
When $L \in [q-1]$, a linear $[n,k,d]$ code is $(\tau,L)$-list decodable
only if $\tau < d$.
Indeed, when $\tau \ge d$, all~$q$ scalar multiples of
any nonzero codeword belong
to (the trivial coset)~$\code$ thereby implying that $L \ge q$.\qed
\end{remark}

It turns out that the condition $L \in [q-1]$ is not too limiting:
we state (in Lemma~\ref{lem:L<q-1} in Section~\ref{sec:bounds})
that when the code rate is bounded away from~$0$ and from~$1$,
the inequality~(\ref{eq:onlyMDS}) can hold only if $L < q-1$
(and, by Theorem~\ref{thm:onlyMDS}, the code is then necessarily MDS),
unless~$L$ is exponentially large in~$n$.
That may justify our focus on (linear) MDS codes in this work.

The next theorem, which we also prove in Section~\ref{sec:bounds},
\fi
establishes the improvement~(\ref{eq:Singleton3})
on~(\ref{eq:Singleton2}) for a wide range of parameters of MDS codes.

\begin{theorem}
\label{thm:SingletonImproved}
Given $L \in \Integers^+$,
let~$\code$ be a linear $[n,k{<}n]$ MDS code over~$F$ and write
\begin{equation}
\ifPAGELIMIT
\else
\label{eq:ur}
\fi
n-k = (L+1)u + r ,
\end{equation}
where $u \in \Integers_{\ge 0}$ and
\ifPAGELIMIT
\else
$r \in [L+1]$.\footnote{%
Namely, $r$ is the ordinary remainder of $n-k$ when divided by $L+1$,
except when $L+1$ divides $n-k$, in which case $r = L+1$.}
\fi
Suppose in addition that
\begin{equation}
\label{eq:Lcondition}
L \le \binom{k-1+u+r}{k-1} .
\end{equation}
Then~$\code$ is $(\tau,L)$-list decodable
\underline{only if} (\ref{eq:Singleton3}) holds.
\end{theorem}

The following corollary presents concrete ranges of parameters
that satisfy the inequality~(\ref{eq:Lcondition})
(and, thus, (\ref{eq:Singleton3}) must hold in these ranges).

\begin{corollary}
\label{cor:SingletonImproved}
Using the notation of Theorem~\ref{thm:SingletonImproved},
if the code~$\code$ therein is $(\tau,L)$-list decodable
then~(\ref{eq:Singleton3}) holds
in any one of the following cases:
\begin{list}{}{\settowidth{\labelwidth}{\textrm{(a)}}}
\itemsep.5ex
\item[(a)]
$\displaystyle k \ge L$,
\item[(b)]
$\displaystyle k \ge 2$ and $\displaystyle n \ge (L+1)h + k$,
where~$h$ is the smallest nonnegative integer that
\ifPAGELIMIT
\else
satisfies\footnote{%
When $k \ge L$, cases~(a) and~(b) coincide.
Otherwise, as~$k$ increases from~$2$ to $L-1$,
the value of~$h$ decreases from $L-2$ to~$1$
and, respectively, the lower threshold, $(L+1)h + k$,
on~$n$ decreases from $L^2 - L$ to $2L$.}
\fi
\[
\binom{k+h}{k-1} \ge L ,
\]
\item[(c)]
$\displaystyle k \ge 2$ and
$\displaystyle n-k \equiv 0, L-1, \;\textrm{or}\; L \pmod{(L+1)}$,
\item[(d)]
$\displaystyle n-k-1 \le L \le \binom{n-1}{k-1}$.
\end{list}
\end{corollary}

\begin{remark}
\label{rem:gst}
The original posted
version of our work contained a weaker statement of
Theorem~\ref{thm:SingletonImproved}%
${} + {}$Corollary~\ref{cor:SingletonImproved},
which only applied to high-rate MDS codes.
Shortly after that version was posted, 
Goldberg \emph{et al.} posted their (independent) work~\cite{GST},
where they show (in Proposition~3.6 therein) that
when $k > (q/(q-1)) (L-1)$, a linear $[n,k]$ code over~$F$
can be $(\tau,L)$-list decodable only if~(\ref{eq:Singleton3}) holds.
Their result applies to non-MDS codes as well;
for MDS codes, their proof can be modified
so that it covers the same range, $k \ge L$, as
in Corollary~\ref{cor:SingletonImproved}(a).\qed
\end{remark}

\ifPAGELIMIT
\else
The $[n,1,n]$ repetition code over~$F$
is excluded from Theorem~\ref{thm:SingletonImproved} when $L > 1$
(see Example~\ref{ex:repetition}), and so is
the $[n,n{-}1,2]$ parity code when $L \ge n$:
obviously, this code is $(\tau{=}1,L{=}n)$-list decodable
and, thus, (\ref{eq:Singleton3}) does not hold.
More generally, any linear $[n,k]$ MDS code is
$(n{-}k,L)$-list decodable for
\[
L = \binom{n}{k} ,
\]
since every subset of $[n]$ of size $n-k$
contains the support of exactly one vector in each coset of the code;
hence, (\ref{eq:Singleton3}) does not hold for these parameters
as well (even when $L < q$).
In contrast, we also prove in Section~\ref{sec:bounds}
the following result.

\begin{theorem}
\label{thm:largeL}
Let~$\code$ be a linear $[n,k{<}n]$ MDS code over~$F$
and let $L \in \Integers^+$ be smaller than $\binom{n}{k}$.
Then~$\code$ is \underline{not} $(n{-}k,L)$-list decodable whenever
\[
q \ge \binom{n}{k+1} .
\]
\end{theorem}

We note that when $L \ge n-k-1$,
\[
\frac{L(n-k)}{L+1} \ge n-k-1 .
\]
Thus, (\ref{eq:Singleton3}) holds when
\[
n-k-1 \le L \le \binom{n}{k} - 1
\]
for any linear $[n,k]$ MDS code over a sufficiently large field
(and---by Corollary~\ref{cor:SingletonImproved}(d)---under
no conditioning on the field size
when $n - k - 1 \le L \le \binom{n-1}{k-1}$).
\fi

\ifPAGELIMIT
\else
\subsection{Higher-order MDS codes}
\label{sec:intro-LMDS}
\fi

\ifPAGELIMIT
\else
In Sections~\ref{sec:MDS} and~\ref{sec:2-MDS},
\fi
we turn to studying properties of
linear codes that attain the bound~(\ref{eq:Singleton3});
we will refer to such codes as $L$-MDS codes.
To let the case $L = 1$ coincide with ordinary linear MDS codes,
we will need a somewhat stronger notion of list decodability,
to be defined next.
For $a \in \Integers^+$, denote by $\Integers/a$
the set of rationals of the form $b/a$
where $b \in \Integers$.

Let~$\code$ be a code in $F^n$. Given $L \in \Integers^+$
and a nonnegative $\tau \in \Integers/(L+1)$, we say that~$\code$
is \emph{strongly-$(\tau,L)$-list decodable} if
for every $\bldy \in F^n$ there are no $L+1$ distinct codewords
$\bldc_0, \bldc_1, \ldots, \bldc_L \in \code$ such that
\[
\sum_{m \in [0:L]} \weight(\bldy-\bldc_m) \le (L+1)\tau ,
\]
where $\weight(\cdot)$ denotes Hamming weight.
When~$\code$ is a linear $[n,k]$ code over~$F$,
it is strongly-$(\tau,L)$-list decodable if
there are no $L+1$ distinct vectors
\ifPAGELIMIT
\else
\[
\blde_0, \blde_1, \ldots, \blde_L \in F^n
\]
\fi
that are in the same coset of~$\code$ within $F^n$ and satisfy
\begin{equation}
\label{eq:strongly}
\sum_{m \in [0:L]} \weight(\blde_m) \le (L+1) \tau .
\end{equation}
In other words, the \emph{average}
(rather than the \emph{maximum}) weight of 
any $L+1$ distinct vectors in the same coset of~$\code$
must exceed~$\tau$.
(This notion of list decodability is referred to in~\cite{GN}
as \emph{average-radius list decodability};
see also~\cite{Blinovskii} and~\cite{GV}.)
Clearly, any strongly-$(\tau,L)$-list decodable
\ifPAGELIMIT
\else
code\footnote{%
The notion of strong $(\tau,L)$-list decodability is defined
(only) when $L \in \Integers^+$ and~$\tau$ is nonnegative
in $\Integers/(L+1)$; thus, 
when we say that a code is (or is not)
strongly-$(\tau,L)$-list decodable,
it will also imply that~$L$ and~$\tau$ are in their valid range.
This convention also applies to ordinary
$(\tau,L)$-list decodability,
in which case $L \in \Integers^+$ and $\tau \in \Integers_{\ge 0}$.}
\fi
is also (ordinarily) $(\lfloor\tau\rfloor,L)$-list decodable.
Theorem~\ref{thm:Singleton} holds also for strong list decodability
(see the proof in~\cite{ST})
and, as we
\ifPAGELIMIT
\else
will show,
\fi
so does Theorem~\ref{thm:onlyMDS},
except that in both theorems~$\tau$ is now in $\Integers/(L+1)$.
Moreover, we have the following (stronger) counterpart of
\ifPAGELIMIT
\else
Theorem~\ref{thm:SingletonImproved}.\footnote{%
In the context of strong list decodability,
the difference between
the bounds~(\ref{eq:Singleton2}) and~(\ref{eq:Singleton3})
becomes more profound (compared to ordinary list decodability),
due to the finer grid of the possible values of~$\tau$.}
\fi

\begin{theorem}
\label{thm:stronglySingletonImproved}
Let~$\code$ be a linear $[n,k]$ MDS code over~$F$
and let $L \in \Integers^+$ be such that
\[
L \le \binom{n-1}{k-1} .
\]
Then~$\code$ is strongly-$(\tau,L)$-list decodable
\underline{only if} $\tau$ satisfies~(\ref{eq:Singleton3}).
\end{theorem}

A (linear) \emph{$L$-MDS code}
is a linear $[n,k]$ code over~$F$ which is
strongly-$(\tau,L)$-list decodable for $\tau = L(n-k)/(L+1)$.
In such codes, the sum of weights of the $L+1$ lightest
vectors in each coset must exceed $L(n-k)$.
Note that the $1$-MDS property coincides with
the ordinary MDS property.

Section~\ref{sec:MDS} is devoted to
\ifPAGELIMIT
\else
proving some basic properties of $L$-MDS codes.
We state the counterpart of Theorem~\ref{thm:onlyMDS}
for strong list decodability, namely, that
when $L < q$, every $L$-MDS code over~$F$ is MDS
(Theorem~\ref{thm:stronglyonlyMDS}).
We then establish the following closure property
for a wide range of values of~$L$:
if a linear $[n,k]$ MDS code is $L$-MDS,
then it is $\ell$-MDS for smaller values of~$\ell$ as well
(Theorem~\ref{thm:nesting}).
We also show that when $L \ge \binom{n}{k} - k(n-k)$,
every linear $[n,k]$ MDS code is $L$-MDS
(Theorem~\ref{thm:stronglyhighL}).

In Section~\ref{sec:2-MDS}, we concentrate on the case $L = 2$
and prove a necessary and sufficient condition
for a linear $[n,k]$ MDS code to be $2$-MDS
(and, thus, $(\lfloor (2(n-k)/3 \rfloor,2)$-list decodable),
in terms of certain properties of
its punctured codes (Theorem~\ref{thm:puncturingL=2}).
We also show that the $2$-MDS property
is preserved under duality (Theorem~\ref{thm:dualityL=2}):
interestingly, while this property is known to hold for
(ordinary) $1$-MDS codes, it does not generalize to $3$-MDS codes.

\fi
Finally, we present in Section~\ref{sec:GRS} explicit constructions of 
GRS codes that are $2$-MDS.
These constructions improve on the results of~\cite{ST}
in that they apply to smaller fields:
the field size can be polynomial in~$n$
for any fixed~$k$ or $n-k$.

\ifPAGELIMIT
\else
\subsection{Notation}
\label{sec:intro-notation}
\fi

We introduce the following notation.
For a vector $\blde = (e_j)_{j \in [n]} \in F^n$,
we denote by $\Support(\blde)$ the support of~$\blde$.
For a vector $\blde \in F^n$ and a subset $J \subseteq [n]$,
we let $(\blde)_J$ be the subvector of~$\blde$
which consists of the coordinates that are indexed by~$J$.
This notation extends to $\rho \times n$ matrices $H$ over~$F$,
with $(H)_J$ being the $\rho \times |J|$ submatrix of~$H$
consisting of the columns that are indexed by~$J$.

Given a $\rho \times n$ matrix~$H$ over~$F$
and $L+1$ subsets $J_0, J_1, \ldots, J_L$ of $[n]$,
we define the matrix $M = M_{J_0,J_1,\ldots,J_L}(H)$ by
\begin{equation}
\label{eq:M}
M =
\left(
\arraycolsep0.8ex
\renewcommand{\arraystretch}{1.7}
\begin{array}{c|c|c|c|c}
-(H)_{J_0} & (H)_{J_1} & & & \\
\hline
-(H)_{J_0} & & (H)_{J_2} & & \\
\hline
\vdots  & & & \;\;\,\ddots\,\;\; & \\
\hline
-(H)_{J_0} & & & & (H)_{J_L} \\
\end{array}
\right) ,
\end{equation}
where empty blocks denote all-zero submatrices.
\ifPAGELIMIT
\else
The matrix~$M$ has $L \rho$ rows
and $\sum_{m \in [0:L]} |J_m|$ columns
(the matrix~$H$ will typically be taken as
an $(n-k) \times n$ parity-check matrix of
a linear $[n,k]$ code over~$F$).

\begin{remark}
\label{rem:M}
While the role of $J_0$ in~(\ref{eq:M})
may seem to differ from that of the rest of the subsets $J_m$,
in the uses of the matrix~$M$ in this work,
the $L+1$ subsets will enjoy full symmetry.
For example, when $L \rho = \sum_{m \in [0:L]} |J_m|$,
the matrix~$M$ is square and its
determinant is the same as that of
the $((L+1)\rho) \times ((L+1)\rho)$ matrix
\begin{equation}
\label{eq:Malt}
\left(
\arraycolsep0.8ex
\renewcommand{\arraystretch}{1.7}
\begin{array}{c|c|c|c|c}
\,I_\rho\, & (H)_{J_0} & & & \\
\hline
\,I_\rho\, & & (H)_{J_1} & & \\
\hline
\vdots  & & & \;\;\,\ddots\,\;\; & \\
\hline
\,I_\rho\, & & & & (H)_{J_L}
\end{array}
\right)
\end{equation}
(with $I_\rho$ standing for the $\rho \times \rho$ identity matrix):
the symmetry among the subsets $J_m$ is apparent
in~(\ref{eq:Malt}).\qed
\end{remark}

\section{Proofs of bounds}
\label{sec:bounds}

\begin{proof}[Proof of Theorem~\ref{thm:onlyMDS}]
Suppose that~$\code$ is a linear $[n,k,d{\le}n{-}k]$ 
(non-MDS) code over~$F$
and let~$\bldc$ be a codeword of~$\code$ of weight $d$;
without loss of generality we may assume that~$\bldc$ takes
the form
\[
\bigl(
\underbrace{1 \, 1 \, \ldots \, 1}_{d \; \textrm{times}}
 \, 0 \, 0 \, \ldots \, 0 \bigr)
\; \left( \in  F^n \right) .
\]
Let~$Y$ be a set of $L+1 \; (\le q)$ distinct elements of~$F$
and consider a vector $\bldy \in F^n$ in which each element
of~$Y$ appears at least $\lfloor d/(L+1) \rfloor$ times
among the first~$d$ entries of~$\bldy$,
while the remaining entries of~$\bldy$ are all zero. 
Thus, there exist $L+1$ distinct scalar multiples
$\bldc_0, \bldc_1, \ldots, \bldc_L$ of~$\bldc$,
each coinciding with~$\bldy$ on at least
$\lfloor d/(L+1) \rfloor$ out of the first~$d$ coordinates, namely,
$\weight(\bldy-\bldc_m) \le d - \lfloor d/(L+1) \rfloor$
for each $m \in [0:L]$.
Hence, $\code$ is $(\tau,L)$-list decodable only when
\[
\tau \le d - \left\lfloor \frac{d}{L+1} \right\rfloor - 1
< \frac{L \cdot d}{L+1} \le \frac{L(n-k)}{L+1} .
\]
\end{proof}

The next lemma (which we prove in Appendix~\ref{sec:skipped})
shows that the condition $L \in [q-1]$ in Theorem~\ref{thm:onlyMDS}
holds in most cases of interest.
Specifically, the lemma applies to any (not necessarily linear) code
$\code \subseteq F^n$ whose rate,
$R = (\log_q |\code|)/n$, is bounded away from~$0$ and from~$1$,
and to~$L$ whose growth rate (with~$n$)
is bounded from above by some constant.
The statement of the lemma makes use of the function
$\eta_q : [0,1/2] \rightarrow \Realfield$
which is defined for every $\varepsilon \in [0,1/2]$ by
\begin{equation}
\label{eq:eta}
\eta_q(\varepsilon)
= \entropy \left( \frac{q{-}1}{q} (1{-}\varepsilon) \right)
- (1{-}\varepsilon) \cdot \entropy(1/q) ,
\end{equation}
where $\entropy(x) = -x \log_2 x - (1{-}x) \log_2 (1{-}x)$.
(In the proof of the lemma we show that $\eta_q(\varepsilon)$
is bounded from below by the positive
value $\eta_2(\varepsilon)$, which is independent of~$q$.)

\begin{lemma}
\label{lem:L<q-1}
Given a fixed real $\varepsilon \in (0,1/2]$,
let $\code \subseteq F^n$ be a code
of rate $R \in [\varepsilon,1{-}\varepsilon]$,
let $L \in \Integers^+$ be such that
\begin{equation}
\label{eq:L<q-1L}
L < \frac{1}{\sqrt{2n}} \cdot 2^{\eta_q(\varepsilon) \cdot n} ,
\end{equation}
and let $\tau \in \Integers_{\ge 0}$ be such that
\begin{equation}
\label{eq:L<q-1tau}
\tau \ge \frac{L n(1-R)}{L+1} .
\end{equation}
Then~$\code$ is $(\tau,L)$-list decodable
\underline{only if} $L < q-1$.
\end{lemma}

Combining Theorem~\ref{thm:onlyMDS} with
Lemma~\ref{lem:L<q-1} yields the following corollary.

\begin{corollary}
\label{cor:L<q-1}
Given $\varepsilon \in (0,1/2]$,
let~$\code$ be a linear $[n,k]$ code over~$F$
of rate $k/n \in [\varepsilon,1{-}\varepsilon]$,
let $L \in \Integers^+$ satisfy~(\ref{eq:L<q-1L}),
and let $\tau \in \Integers_{\ge 0}$ satisfy~(\ref{eq:onlyMDS}).
Then~$\code$ is $(\tau,L)$-list decodable
\underline{only if} $\code$ is MDS.
\end{corollary}

\begin{remark}
\label{rem:L<q-1}
As $q \rightarrow \infty$, the value of $\eta_q(\varepsilon)$ approaches
$\entropy(\varepsilon)$, which is
the growth rate of $\binom{n}{\varepsilon n}$
(assuming that $\varepsilon n$ is an integer).
Thus, in terms of growth rates, the requirement~(\ref{eq:L<q-1L})
in this case is tight, since we have seen that
linear $[n,k{=}\varepsilon n]$ MDS codes are
$(\tau,L)$-list decodable for $L = \binom{n}{k}$
and $\tau = n-k > L(n-k)/(L+1)$, regardless of~$q$
(as long as MDS codes exist over~$F$, e.g., when $q \ge n-1$).\qed
\end{remark}

\begin{proof}[Proof of Theorem~\ref{thm:SingletonImproved}]
Let
\begin{equation}
\label{eq:tau}
\tau = \left\lceil \frac{L(n-k) + 1}{L+1} \right\rceil
\stackrel{\textrm{(\ref{eq:ur})}}{=}
\left\lceil \frac{L((L+1)u + r) + 1}{L+1} \right\rceil
= L u + r
\end{equation}
and note that~$\tau$ is the smallest integer that is greater
than the right-hand side of~(\ref{eq:Singleton3}).
We construct $L+1$ distinct vectors
$\blde_0, \blde_1, \ldots, \blde_L \in F^n$
of weight${} \le \tau$ which belong
to the same coset of~$\code$;
this will imply that~$\code$ is not $(\tau,L)$-list decodable.

We recall that 
every subset of $[n]$ of size $n-k+1$ is a support of $q-1$
codewords of minimum weight $n-k+1$ of the MDS code~$\code$,
and those $q-1$ codewords
are scalar multiples of each other~\cite[Ch.~11, \S 3]{MS}. 
Denoting
\begin{equation}
\label{eq:s}
s = k + u + r
\end{equation}
(which, by~(\ref{eq:ur}), is in $[k+1:n]$), there are
\[
\binom{s-1}{k-1} \stackrel{\textrm{(\ref{eq:Lcondition})}}{\ge} L
\]
codewords of weight $n-k+1$ that have a~$1$ as their first coordinate
and their $k-1$ zero entries are all located within
the first~$s$ coordinates of the codeword.
Let $\bldc_1, \bldc_2, \ldots, \bldc_L$ be~$L$ such codewords.

We turn to defining the $L+1$ vectors $\blde_m$.
To this end, we partition $[s+1:n]$, which is of size
\begin{eqnarray}
n - s
& \stackrel{\textrm{(\ref{eq:s})}}{=} &
n - k - u - r
\nonumber \\
& \stackrel{\textrm{(\ref{eq:ur})}}{=} &
((L+1) u + r) - u - r
\nonumber \\
\label{eq:n-s}
& = &
L u ,
\end{eqnarray}
into~$L$ distinct subsets $T_m$, $m \in [L]$,
each of size $u$.
The vector $\blde_0 \in F^n$ is defined by
\[
\begin{array}{lcl}
(\blde_0)_{[s]} & = & (1 \, 0 \, 0 \, \ldots \, 0) \\
(\blde_0)_{T_m} & = & (\bldc_m)_{T_m} , \quad m \in [L] 
\end{array} ,
\]
and the remaining~$L$ vectors are defined by
\begin{equation}
\label{eq:othervectors}
\blde_m = \blde_0 - \bldc_m , \quad m \in [L] .
\end{equation}
Clearly, the $L+1$ vectors $\blde_m$ are distinct
and they all belong to the same coset of~$\code$.

We next show that $\weight(\blde_m) \le \tau$
for every $m \in [0:L]$. For $m = 0$ we have
\[
\weight \left( \blde_0 \right)
\le
1 + n - s
\stackrel{\textrm{(\ref{eq:n-s})}}{\le}
L u + 1
\stackrel{\textrm{(\ref{eq:tau})}}{\le} \tau ,
\]
and for $m \in [L]$ we have
\begin{eqnarray*}
\weight \left( \blde_m \right)
& = &
{\underbrace{\weight\left ( (\blde_m)_{[s]} \right)}_{s-k}}
+
{\underbrace{\weight \left( (\blde_m)_{[s+1:n]} \right)}_%
   {{} \le n-s-|T_m|}} \\
& \le &
(s - k) + (n - s - u) \\
& = &
n - k - u
\stackrel{\textrm{(\ref{eq:ur})}}{=}
L u + r
\stackrel{\textrm{(\ref{eq:tau})}}{=} \tau .
\end{eqnarray*}
\end{proof}

\begin{proof}[Proof of Corollary~\ref{cor:SingletonImproved}]
We show that the inequality~(\ref{eq:Lcondition})
holds in each of the cases (b)--(d)
(with case~(b) becoming case~(a) when $h = 0$).

\emph{(b)}
   From $n-k \ge (L+1) h$ we get $u + r \ge h+1$
and, so, the right-hand side of~(\ref{eq:Lcondition}) satisfies
\[
\binom{k-1+u+r}{k-1}
\ge \binom{k+h}{k-1} \ge L .
\]

\emph{(c)}
Here $r \ge L-1$ and, so, $k - 1 + u + r \ge L$
and the right-hand side of~(\ref{eq:Lcondition})
is at least $\binom{L}{k-1} \ge L$.

\emph{(d)}
When $n - k - 1 \le L$
we have $u = 0$ and $r = n - k$
and, so, the right-hand side of~(\ref{eq:Lcondition})
equals $\binom{n-1}{k-1}$.
\end{proof}

\begin{proof}[Proof of Theorem~\ref{thm:largeL}]
Let~$H$ be an $(n-k) \times n$ parity-check matrix of~$\code$.
There are at most $V_q(n,n-k-1)$ column vectors in $F^{n-k}$
that can be written as linear combinations of less than $n-k$
columns of $H$. By our assumptions,
\begin{eqnarray*}
V_q(n,n-k-1)
& < &
\binom{n}{n-k-1} \cdot q^{n-k-1} \\
& = &
\binom{n}{k+1} \cdot q^{n-k-1} \le q^{n-k}
\end{eqnarray*}
(except when $k = n-1$, where the first inequality is weak and
the second is strict; in fact, for this case the theorem holds
without any restrictions on~$q$).
Hence, there is a column vector $\blds \in F^{n-k}$
that can be written as a proper linear combination of the columns
of $(H)_J$, for any subset $J \subseteq [n]$ of size $n-k$;
namely, each column of $(H)_J$ appears with a nonzero coefficient
in the linear combination. In particular,
for distinct subsets~$J$ we get different linear combinations.
We conclude that~$\code$ cannot be $(n{-}k,L)$-list decodable
when $L < \binom{n}{n-k} = \binom{n}{k}$.
\end{proof}

\begin{remark}
\label{rem:largeL}
If we remove any conditioning on the field size,
then there do exist MDS codes
that are $(n{-}k,L)$-list decodable for some $L < \binom{n}{k}$
(by Corollary~\ref{cor:SingletonImproved}(d),
such $L$ should also satisfy $L > \binom{n-1}{k-1}$).
For example, the linear $[8,4]$ code over $\GF(7)$
in Example~\ref{ex:nonLMDS} below is $(4,50)$-list decodable,
and the linear $[6,2]$ code over $\GF(5)$
in Example~\ref{ex:LMDSk=2,L=3} below
is $(4,11)$-list decodable.\qed
\end{remark}

\begin{proof}[Proof of Theorem~\ref{thm:stronglySingletonImproved}]
We simplify the proof of Theorem~\ref{thm:SingletonImproved},
redefining~$s$ therein to be~$n$ (instead of~(\ref{eq:s}));
respectively, the subsets $T_m$ become empty
and $\blde_0$ becomes $(1 \, 0 \, 0 \, \ldots \, 0)$. 
Keeping the definition of
$\blde_1, \blde_2, \ldots, \blde_L$
as in~(\ref{eq:othervectors}),
each has weight $n-k$ and, therefore,
the left-hand side of~(\ref{eq:strongly}) equals $L(n-k) + 1$.
Hence, $\code$ is strongly-$(\tau,L)$-list decodable
only if $\tau \le L(n-k)/(L+1)$.
\end{proof}
\fi

\ifPAGELIMIT
\else
\section{Generalizations of the MDS property}
\label{sec:MDS}

\subsection{$L$-MDS codes}
\label{sec:LMDS}

Recall the definitions of strong list decodability
and $L$-MDS codes from Section~\ref{sec:intro-LMDS}.
Specifically,
an $L$-MDS code is a linear $[n,k]$ code over~$F$ which is
strongly-$(\tau,L)$-list decodable for $\tau = L(n-k)/(L+1)$.
Clearly, such a code is also
$(\lfloor \tau \rfloor,L)$-list decodable
(in the ordinary sense),
but the converse is generally not true.
The next example presents a code which is
$(\lceil \tau \rceil,L)$-list decodable
yet is not strongly-$(\tau,L)$-list decodable.

\begin{example}
\label{ex:nonLMDS}
Let~$\code$ be the linear $[8,4]$ code over $F = \GF(7)$
with a parity-check matrix
\[
H =
\left(
\begin{array}{cccccccc}
1 & 1 & 1 & 1 & 1 & 1 & 1 & 0 \\
0 & 1 & 2 & 3 & 4 & 5 & 6 & 0 \\
0 & 1 & 4 & 2 & 2 & 4 & 1 & 0 \\
0 & 1 & 1 & 6 & 1 & 6 & 6 & 1
\end{array}
\right) .
\]
This code is a doubly-extended Reed--Solomon code
and is therefore MDS~\cite[p.~323]{MS}.
By exhaustively checking the cosets of the code
we find that this code is
(ordinarily) $(\tau{=}4,L{=}50)$-list decodable;
note that for these parameters,
\[
4 = \tau
= \left\lceil \frac{L(n-k)}{L+1} \right\rceil
> \frac{L(n-k)}{L+1} = \frac{200}{51}
\]
(and, so, the bound~(\ref{eq:Singleton3}) is exceeded in this case).
On the other hand, the coset that contains
the vector $(0 \, 0 \, 0 \, 0 \, 0 \, 2 \, 6 \, 4)$
has a weight distribution as shown in Table~\ref{tab:nonLMDS},
with $A_w$ standing for the number of vectors of weight~$w$
in the coset.
\begin{table}[hbt]
\label{tab:nonLMDS}
\caption{Weight distribution of a coset of the code
in Example~\protect\ref{ex:nonLMDS}.}
\normalsize
\[
\begin{array}{ccccccccc}
\hline
\hline
A_0 & A_1 & A_2 & A_3 & A_4 & A_5 & A_6 & A_7 & A_8 \\
\hline
0   & 0   & 0   & 5   & 45  & 162 & 566 & 921 & 702 \\
\hline
\hline
\end{array}
\]
\end{table}
It follows that the sum of weights of the lightest $L+1 = 51$
vectors in the coset is
\[
3 \cdot A_3 + 4 \cdot A_4 + 5 \cdot 1 = 200 = L(n-k) ,
\]
which means that~$\code$ is not
strongly-$(200/51,50)$-list decodable
and, therefore, is not $50$-MDS
(however, this code turns out to be $L$-MDS for every $L \ge 51$).\qed
\end{example}

We next present some properties of $L$-MDS codes,
starting with the following theorem, which
\fi
is the counterpart of Theorem~\ref{thm:onlyMDS}
for strong list decodability.

\begin{theorem}
\label{thm:stronglyonlyMDS}
Given $L \in [q-1]$, every $L$-MDS code over~$F$ is MDS.
\end{theorem}

\ifPAGELIMIT
\else
\begin{proof}
If~$\code$ is not MDS, then the left-hand side of~(\ref{eq:strongly})
would not exceed $L(n-k)$ when the $L+1 \; (\le q)$ vectors $\blde_m$
are taken as $L+1$ distinct scalar multiples of
a codeword of weight${} \le n-k$, with one of the scalars being zero.
\end{proof}
\fi

We also have the following closure property.

\begin{theorem}
\label{thm:nesting}
Let~$\code$ be a linear $[n,k]$ MDS code over~$F$
and let $L \in \Integers^+$ be such that
\begin{equation}
\ifPAGELIMIT
\else
\label{eq:nesting}
\fi
L < \max \left\{ \binom{n-1}{k-1}
\biggm/ \binom{\lceil (n+k)/2 \rceil - 2}{k-1} ,
k \right\} + 1 .
\end{equation}
If~$\code$ is $L$-MDS,
then it is also $\ell$-MDS for every $\ell \in [L]$.
\end{theorem}

\ifPAGELIMIT
\else
We prove the theorem using the following lemma,
which is proved in Appendix~\ref{sec:skipped}.

\begin{lemma}
\label{lem:nesting}
Given $w, \ell, s, t \in \Integers^+$,
let $J_1, J_2, \ldots, J_\ell$ be subsets of $[w]$,
each of size at least~$s$. If
\[
\ell < \max \left\{ \binom{w}{t} \biggm/ \binom{w-s}{t} , t+1 \right\} ,
\]
then there exists a subset $X \subseteq [w]$
of size at most~$t$ that intersects with $J_m$,
for each $m \in [\ell]$.
\end{lemma}

\begin{proof}[Proof of Theorem~\ref{thm:nesting}]
The proof is by contradiction.
Suppose that ($\code$ is $L$-MDS and)
$\ell$ is the largest in $[L]$ for which~$\code$ is not $\ell$-MDS,
namely, there exist~$\ell$ distinct vectors
$\blde_0, \blde_1, \ldots, \blde_\ell \in F^n$
such that
\begin{equation}
\label{eq:nesting1}
\sum_{m \in [0:\ell]} \weight(\blde_m) \le \ell(n-k)
\end{equation}
and, for an $(n-k) \times n$
parity-check matrix~$H = (\bldh_j)_{h \in [n]}$ of~$\code$:
\begin{equation}
\label{eq:nesting2}
\blds = H \blde_0^\top = H \blde_1^\top = \cdots
= H \blde_\ell^\top .
\end{equation}
We note that~$\blds$ is nonzero
and, therefore, each vector $\blde_m$ is nonzero:
if~$\blds$ were zero then at least~$\ell$ of the vectors $\blde_m$
would be nonzero codewords of~$\code$,
in which case the left-hand side of~(\ref{eq:nesting1}) would be
at least $\ell(n{-}k{+}1)$.
We also note that since the difference between any two distinct
vectors~$\blde_m$ is a nonzero codeword of~$\code$,
we have $\weight(\blde_m) \ge d/2 = (n-k+1)/2$
for all $m \in [0:\ell]$ except, possibly, for one index~$m$,
say $m = 0$. We also assume without loss of generality
that the last entry of $\blde_0$ is nonzero.

Next, we construct a subset $J \subseteq [n]$ which
intersects with $\Support(\blde_m)$, for each $m \in [0:\ell]$.
Specifically, we let $J = \{ n \} \cup X$, where~$X$
is the subset guaranteed by Lemma~\ref{lem:nesting}
when applied with
$w = n-1$,
$\ell \; (\le L - 1)$,
$s = \lceil d/2 \rceil = \lfloor (n-k)/2 \rfloor + 1$,
$t = k-1$,
and $J_m = \Support(\blde_m)$, $m \in [\ell]$.
By the lemma we then get that $|J| \le k$.
Writing $J' = [n] \setminus J$,
it follows that the $(n-k) \times |J'|$
submatrix $(H)_{J'} = (\bldh_j)_{j \in J'}$
has full rank $n-k$, which means that there exists
a subset $J_{\ell+1} \subseteq J'$ of size $|J_{\ell+1}| \le n-k$
such that
\[
\blds = \sum_{j \in J_{\ell+1}} a_j \bldh_j ,
\]
for some $\blda = (a_j)_{j \in J_{\ell+1}}$ over~$F$.
Define the vector $\blde_{\ell+1} \in F^n$ by
$(\blde_{\ell+1})_{J_{\ell+1}} = \blda$
and $(\blde_{\ell+1})_{[n] \setminus J_{\ell+1}} = \bldzero$.
We have
\begin{equation}
\label{eq:nesting3}
\weight(\blde_{\ell+1}) \le |J_{\ell+1}| \le n-k
\quad \textrm{and} \quad
\blds = H \blde_{\ell+1}^\top .
\end{equation}
Moreover, for each $m \in [0:\ell]$ there exists a coordinate on which
$\blde_{\ell+1}$ is zero
while $\blde_m$ is not, namely, $\blde_{\ell+1} \ne \blde_m$.
Combining (\ref{eq:nesting1})--(\ref{eq:nesting2})
with~(\ref{eq:nesting3}) yields that~$\code$ is not $(\ell+1)$-MDS,
thereby contradicting our assumption that~$\ell$
is the largest in $[L]$ for which~$\code$ is not $\ell$-MDS.
\end{proof}

\begin{remark}
\label{rem:nesting}
Fixing the rate $R = k/n$ to be bounded away from~$0$ and from~$1$,
the expression in the right-hand side of~(\ref{eq:nesting})
can be shown to grow exponentially with~$n$.\qed
\end{remark}
\fi

\ifPAGELIMIT
\else
There are few cases where the MDS property implies
the $L$-MDS property for $L > 1$, as shown
in the following examples and in
Theorem~\ref{thm:stronglyhighL} below.

\begin{example}
\label{ex:LMDSn-k=2}
We verify that when $k \ge n-2$,
any linear $[n,k]$ MDS code~$\code$ over~$F$ is $L$-MDS 
for every $L \in \Integers^+$.
Let $\blde_0, \blde_1, \ldots, \blde_L$ be distinct
vectors in the same coset of~$\code$. If all of them
have weight~$2$ or more, then their sum of weights 
is at least $2(L+1) \ge L(n-k) + 2$. Otherwise, if, say,
$\weight(\blde_0) \le 1$ then
$\weight(\blde_m) \ge n-k+1 - \weight(\blde_0)$ for all $m \in [L]$
and, so,
\[
\sum_{m \in [0:L]} \weight(\blde_m)
\ge L(n-k+1) - (L-1) \cdot
\underbrace{\weight(\blde_0)}_{{} \le 1}
\ge L(n-k) + 1
\]
(where the last inequality is strict when $\blde_0 = \bldzero$);
i.e., $\code$ is $L$-MDS.
Considering the case $\weight(\blde_0)= 1$,
it follows from the proof of
Theorem~\ref{thm:stronglySingletonImproved}
(and of Theorem~\ref{thm:SingletonImproved})
that in this case, there are
exactly $\binom{n-1}{k-1}$ vectors of weight $n-k$
in the coset that contains $\blde_0$,
and all the remaining vectors in that coset have larger weight.
Hence, when $L > \binom{n-1}{k-1}$
(namely, when we are outside the range of~$L$ to
which Theorem~\ref{thm:stronglySingletonImproved} applies),
we have
\[
\sum_{m \in [0:L]} \weight(\blde_m) =
\underbrace{\weight(\blde_0)}_{1}
+ \sum_{m \in [L]} \weight(\blde_m) > 1 + L(n-k) .
\]
We conclude that for linear $[n,k{\ge}n{-}2]$ MDS codes,
the bound~(\ref{eq:Singleton3})
is exceeded when $L > \binom{n-1}{k-1}$.\qed
\end{example}

\begin{example}
\label{ex:LMDSk=1}
The $[n,k{=}1]$ repetition code over~$F$ is $L$-MDS 
for every $L \in \Integers^+$:
given $L+1$ distinct vectors $\blde_0, \blde_1, \ldots, \blde_L$
in the same coset of the code, 
write $J_m = \Support(\blde_m)$ and let
$J'_m = [n] \setminus J_m$.
Since nonzero codewords have weight~$n$,
the sets $J'_m$ must be disjoint and, so,
\begin{eqnarray*}
\sum_{m \in [0:L]} \weight(\blde_m)
& = &
\sum_{m \in [0:L]} \left(n-|J'_m| \right) \\
& = &
(L+1) n - \sum_{m \in [0:L]} |J'_m| \\
& \ge &
(L+1) n - n \ge L(n-k) + 1 ,
\end{eqnarray*}
with the last inequality being strict when $L > 1$.
Hence, $\code$ is $L$-MDS 
and, in addition, (\ref{eq:Singleton3}) is exceeded
when we are outside the range of
Theorem~\ref{thm:stronglySingletonImproved}.\qed
\end{example}

\begin{example}
\label{ex:LMDSk=2,L=2}
We show that every linear $[n,k{=}2]$ MDS code~$\code$ is $2$-MDS.
Let $\blde_0, \blde_1, \blde_2$
be distinct vectors in the same coset of~$\code$
and write $J_m = \Support(\blde_m)$
and $J'_m = [n] \setminus J_m$.
Since nonzero codewords have weight${} \ge n-k-1 = n-1$,
we have $|J_m \cup J_\ell| \ge n-1$ for $0 \le m < \ell \le 2$
and, so,
\[
\left| J'_m \cap J'_\ell \right|
=
\left| (J_m \cup J_\ell)' \right|
\le 1 .
\]
   From $\left| \bigcup_{m \in [0:2]} J'_m \right| \le n$
we get, by the inclusion--exclusion principle:
\begin{eqnarray*}
\lefteqn{
\sum_{m \in [0:2]} \left| J'_m \right|
} \makebox[0ex]{} \\
& = &
\Bigl| \! \bigcup_{m \in [0:2]} \! J'_m \Bigr|
+
\sum_{0 \le m < \ell \le 2}
\left| J'_m \cap J'_\ell \right|
-
\Bigl| \! \bigcap_{m \in [0:2]} \! J'_m \Bigr| \\
& \le &
n + 3.
\end{eqnarray*}
Hence, for $L = k = 2$:
\[
\sum_{m \in [0:2]} \weight(\blde_m)
= \sum_{m \in [0:2]} \left( n - \left| J'_m \right| \right)
\ge 2n - 3 > L(n-k) . \qed
\]
\end{example}

Contrary to Examples~\ref{ex:LMDSn-k=2} and~\ref{ex:LMDSk=1}, however,
we cannot claim that the codes in the last example are $L$-MDS
for $L > 2$. We demonstrate this in the next example.

\begin{example}
\label{ex:LMDSk=2,L=3}
Let~$\code$ be the linear $[6,2]$ MDS code over $F = \GF(5)$
with the generator matrix
\[
G =
\left(
\begin{array}{cccccc}
1 & 0 & 1 & 1 & 1 & 1 \\
0 & 1 & 1 & 2 & 3 & 4
\end{array}
\right) .
\]
The following four vectors can be verified to be
in the same coset of $\code$ within $F^6$:
\[
( 1 \, 1 \, 3 \, 0 \, 0 \, 0 ) , \;
( 0 \, 4 \, 0 \, 0 \, 3 \, 1 ) , \;
( 4 \, 0 \, 0 \, 1 \, 0 \, 4 ) , \;
( 0 \, 0 \, 1 \, 2 \, 1 \, 0 ) .
\]
Therefore, $\code$ is not $(3,3)$-list decodable and, hence,
is not $3$-MDS. In fact, $\code$ is not $L$-MDS also when
$L = 4, 5, 6$, yet it is $L$-MDS for $L \ge 7$.\qed
\end{example}

With every MDS code~$\code$ we can associate a threshold, $L_0(\code)$,
which is the smallest positive integer such that~$\code$ is
$L$-MDS when $L \ge L_0(\code)$
(such a threshold always exists since every code is $L$-MDS
for $L \ge |\code|$).
We provide an upper bound on $L_0(\code)$
in the next theorem, which we prove using the next two lemmas:
the first presents a property of the weight distribution
of a coset of an MDS code and follows from MacWilliams' identities,
and the second is proved in Appendix~\ref{sec:skipped}.

\begin{lemma}[\cite{Bonneau}]
\label{lem:Bonneau}
Let~$\code$ be a linear $[n,k]$ MDS code over~$F$
and let $(A_w)_{w \in [0:n]}$ be the weight distribution
of some coset~$X$ of~$\code$ within $F^n$
(where $A_w$ is the number of vectors of weight~$w$ in~$X$).
Then
\[
\sum_{w=0}^{n-k} \binom{n-w}{k} A_w = \binom{n}{k} .
\]
\end{lemma}

\begin{lemma}
\label{lem:vartheta}
Given positive integers $k < n$, define the rational sequence
$(\vartheta_w)_{w \in [0:n-k-1]}$ by
\begin{equation}
\label{eq:vartheta}
\vartheta_w =
\frac{1}{n{-}k{-}w} \left( \binom{n-w}{k} - (n{-}k{-}w{+}1) \right) .
\end{equation}
This sequence is all-zero when $k = 1$
and is strictly decreasing when $k > 1$.
\end{lemma}
\fi

\begin{theorem}
\label{thm:stronglyhighL}
Let~$\code$ be a linear $[n,k]$ MDS code over~$F$.
Then $\code$ is $L$-MDS for every $L \in \Integers^+$ such that
\begin{equation}
\ifPAGELIMIT
\else
\label{eq:stronglyhighL1}
\fi
L \ge \binom{n}{k} - k(n-k) .
\end{equation}
\end{theorem}

\ifPAGELIMIT
\else
\begin{proof}
Fix~$X$ to be any coset of~$\code$ within $F^n$
and let $(A_w)_{w \in [0:n]}$ be the weight distribution of~$X$.
We distinguish between two (disjoint) cases.

\emph{Case 1:}\/ the weight distribution satisfies the inequality
\begin{equation}
\label{eq:stronglyhighL2}
\sum_{w=0}^{n-k-1} (n{-}k{-}w) A_w \le n-k-1 .
\end{equation}
For a vector $\blde \in X$, define its \emph{deficiency} by
\[
\delta(\blde) = \max \{ n-k-\weight(\blde), 0 \} .
\]
Observing that the left-hand side of~(\ref{eq:stronglyhighL2})
is the sum of the deficiencies of all the vectors in~$X$,
we then get that for any $L+1$ distinct vectors
$\blde_0, \blde_1, \ldots, \blde_L \in X$:
\begin{equation}
\label{eq:stronglyhighL3}
\sum_{m \in [0:L]} \delta(\blde_m)
\le \sum_{\blde \in X} \delta(\blde) \le n-k-1 .
\end{equation}
Hence,
\begin{eqnarray*}
\sum_{m \in [0:L]} \weight(\blde_m)
& \ge &
\sum_{m \in [0:L]} (n-k - \delta(\blde_m)) \\
& = & (L+1)(n-k) - \sum_{m \in [0:L]} \delta(\blde_m) \\
& \stackrel{\textrm{(\ref{eq:stronglyhighL3})}}{\ge} & 
(L+1)(n-k) - (n-k-1) \\
& = & L(n-k) + 1
\end{eqnarray*}
(regardless of whether~$L$ satisfies~(\ref{eq:stronglyhighL1})).

\emph{Case 2:}\/ the weight distribution satisfies the inequality
\begin{equation}
\label{eq:stronglyhighL4}
\sum_{w=0}^{n-k-1} (n{-}k{-}w) A_w \ge n-k .
\end{equation}
Let $\blde_0, \blde_1, \ldots, \blde_L$ be distinct
vectors in~$X$. Then
\begin{eqnarray*}
\lefteqn{
\sum_{m \in [0:L]} \weight(\blde_m)
} \makebox[5ex]{} \\
& \ge &
\sum_{w=0}^{n-k} w A_w
+ \biggl( L+1- \sum_{w=0}^{n-k} A_w \biggr) (n-k+1) \\
& = &
(L+1)(n-k+1) - \sum_{w=0}^{n-k} (n{-}k{-}w{+}1) A_w  .
\end{eqnarray*}
We seek to show that the last expression is at least $L(n-k) + 1$,
namely, that
\[
L + n-k - \sum_{w=0}^{n-k} (n{-}k{-}w{+}1) A_w  \ge 0 ,
\]
and, to this end, it suffices (by~(\ref{eq:stronglyhighL1}))
to show that
\[
\binom{n}{k} - \sum_{w=0}^{n-k} (n{-}k{-}w{+}1) A_w  \ge (k-1)(n-k)
\]
which, by Lemma~\ref{lem:Bonneau}, is equivalent to
\begin{equation}
\label{eq:stronglyhighL5}
\sum_{w=0}^{n-k-1}
\left( \binom{n-w}{k} - (n{-}k{-}w{+}1) \right) A_w  \ge (k-1)(n-k) .
\end{equation}

For $w \in [0:n{-}k{-}1]$, define
\[
B_w = (n{-}k{-}w) A_w
\]
and let the sequence $(\vartheta_w)_{[0:n-k-1]}$
be as in~(\ref{eq:vartheta}).
Then~(\ref{eq:stronglyhighL4})
and~(\ref{eq:stronglyhighL5}) become, respectively,
\begin{equation}
\label{eq:stronglyhighL6}
\sum_{w=0}^{n-k-1} B_w \ge n-k
\end{equation}
and
\begin{equation}
\label{eq:stronglyhighL7}
\sum_{w=0}^{n-k-1}
\vartheta_w B_w  \ge (k-1)(n-k) .
\end{equation}
By Lemma~\ref{lem:vartheta},
the left-hand side of~(\ref{eq:stronglyhighL7}),
when seen as a function of nonnegative integer $(n{-}k)$-tuples
$(B_w)_{w \in [0:n-k-1]}$ that satisfy~(\ref{eq:stronglyhighL6}),
is minimized when
$B_w = 0$ for $w \in [0:n{-}k{-}2]$ and $B_{n-k-1} = n-k$.
For this choice, we get equality in~(\ref{eq:stronglyhighL7}).
\end{proof}

The range~(\ref{eq:stronglyhighL1}) of~$L$ in
Theorem~\ref{thm:stronglyhighL}
is tight for $k \in \{ 1, n-1, n \}$
(see Examples~\ref{ex:LMDSn-k=2} and~\ref{ex:LMDSk=1})
and for the code in Example~\ref{ex:LMDSk=2,L=3},
yet is loose for $k = n-2$ when $n \ge 4$
(Example~\ref{ex:LMDSn-k=2})
and for the code in Example~\ref{ex:nonLMDS},
and also for the (rather peculiar) codes in the next example.

\begin{example}
\label{ex:GF11GF73}
Let $F = \GF(11)$ and $n = 6$
and, for $k \in [n]$, let $\code_k$ be
the linear $[n,k]$ code over~$F$ with
the $(n-k) \times n$ parity-check matrix
$H = (\alpha_j^i)_{i \in [0:n-k-1],j\in[n]}$, where
\[
(\alpha_j)_{j \in [n]} = (0 \; 1 \; 4 \; 9 \; 5 \; 3) .
\]
This code is a GRS code over~$F$ and is therefore MDS.
An exhaustive check reveals
that $\code_k$ is $L$-MDS for \emph{every} $k \in [n]$
and \emph{every}
$L \in \Integers^+$.\footnote{%
The check is required only for $k \in \{ 2, 3 \}$,
as the remaining values of~$k$ are covered by
Examples~\ref{ex:LMDSn-k=2} and~\ref{ex:LMDSk=1}.
The fact that the elements $\alpha_j$ range over
the quadratic residues of $\GF(11)$---including zero---may
or may not be coincidental.}
The same holds for a similar construction of length $n = 7$ 
over $F = \GF(73)$, where
\[
(\alpha_j)_{j \in [n]} = (0 \; 1 \; 9 \; 8 \; 3 \; 16 \; 34) .
\]\qed
\end{example}
\fi

\ifPAGELIMIT
\else
\subsection{Lightly-$L$-MDS codes}
\label{sec:lightlyLMDS}
\fi

The next definition introduces a weaker notion of
strong list decodability; the notion is quite artificial,
yet it will be useful in our analysis.

Given $L \in \Integers^+$ and
a nonnegative $\tau \in \Integers/(L+1)$,
a linear $[n,k,d]$ code~$\code$ over~$F$ is called
\emph{lightly-$(\tau,L)$-list decodable}
if there are no $L+1$
\ifPAGELIMIT
\else
\emph{nonzero} vectors
\[
\blde_0, \blde_1, \ldots, \blde_L \in \Ball(n,d-1)
\]
\fi
in the same coset of~$\code$ within $F^n$
that have \emph{disjoint supports}
\ifPAGELIMIT
\else
and\footnote{Nonzero vectors
that have disjoint supports are necessarily distinct.
In fact, requiring the vectors to be nonzero
is equivalent to requiring them to be distinct,
since distinct vectors in
$\Ball(n,d-1)$ that belong to the same coset of~$\code$
cannot be codewords of~$\code$ and therefore must all be
nonzero.}
satisfy~(\ref{eq:strongly}).

\begin{remark}
\label{rem:incomparable}
Generally, light list decodability is incomparable with
ordinary list decodability.
Specifically, among any three vectors within the same coset
of a linear $[n,k,d]$ code~$\code$
there is at least one vector of weight${} \ge d/2$
and the other two vectors have total weight${} \ge d$.
Hence, when $d > 2n/3$, 
these three vectors cannot have disjoint supports,
implying that~$\code$
is lightly-$(\tau,3)$-list decodable for any~$\tau$
(but it is not $(\tau,3)$-list decodable
when, say, $\tau = n$ and $|\code| > 3$).
Conversely, the linear $[15,5,4]$ code over $\GF(2)$
which is generated by
\[
G =
\left(
\arraycolsep0.7ex
\renewcommand{\arraystretch}{0.8}
\begin{array}{ccccccccccccccc}
1 & 1 & 1 & 1 &   &   &   &   &   &   &   &   &   &   &   \\
1 & 1 &   &   & 1 & 1 &   &   &   &   &   &   &   &   &   \\
1 & 1 &   &   &   &   & 1 & 1 &   &   &   &   &   &   &   \\
1 & 1 &   &   &   &   &   &   & 1 & 1 &   &   &   &   &   \\
1 & 1 &   &   &   &   &   &   &   &   & 1 & 1 & 1 & 1 & 1
\end{array}
\right)
\]
can be verified exhaustively to be $(\tau{=}3,L{=}5)$-list decodable,
yet it is not lightly-$(3,5)$-list decodable:
denoting by $\bldg_i$ the $i$th row of~$G$ and letting
\[
\blde = (1 \, 1 \, 0 \, 0 \, \ldots 0) ,
\]
the vectors
$\blde, \blde + \bldg_1, \blde + \bldg_2, \ldots, \blde + \bldg_5$
have disjoint supports and total weight
$15 = L \cdot \tau$.\qed
\end{remark}
\fi

In analogy to the previous definition of $L$-MDS codes,
we define a (linear) \emph{lightly-$L$-MDS code}
to be a linear $[n,k]$ code over~$F$ which is
lightly-$(\tau,L)$-list decodable for $\tau = L(n-k)/(L+1)$.
Here, too, the lightly-$1$-MDS property coincides with
the ordinary MDS property.

The following lemma provides a necessary and sufficient condition
for a linear MDS code~$\code$ to be lightly-$L$-MDS.
The condition is expressed in terms of the non-singularity
of matrices $M = M_{J_0,J_1,\ldots,J_L}(H)$
(as in~(\ref{eq:M})) which are computed for
a parity-check matrix~$H$ of~$\code$.

\begin{lemma}
\label{lem:lightly}
Given $L \in \Integers^+$,
let~$\code$ be a linear $[n,k{=}R n]$ MDS code~$\code$ over~$F$ where
\begin{equation}
\ifPAGELIMIT
\else
\label{eq:ratelightly}
\fi
R \ge 1 - \frac{1}{L} ,
\end{equation}
and let~$H$ be an $(n-k) \times n$ parity-check matrix of~$\code$.
Then~$\code$ is lightly-$L$-MDS, if and only if
\[
\det(M_{J_0,J_1,\ldots,J_L}(H)) \ne 0
\]
for every $L+1$ disjoint subsets
\ifPAGELIMIT
\else
\[
J_0, J_1, \ldots, J_L \subseteq [n]
\]
\fi
that satisfy the following conditions:
\begin{list}{}{\settowidth{\labelwidth}{\textit{(S2)}}}
\itemsep.5ex
\item[S1)]
$\displaystyle |J_m| \le n-k$, for $m \in [0:L]$, and---
\item[S2)]
$\displaystyle\sum_{m \in [0:L]} |J_m| = L(n-k)$.
\end{list}
\end{lemma}

\ifPAGELIMIT
\else
\begin{proof}
Starting with the ``if'' part,
suppose that~$\code$ is not lightly-$L$-MDS.
Then there exist nonzero vectors
$\blde_0, \blde_1, \ldots, \blde_L \in \Ball(n,n-k)$
of disjoint supports such that
\begin{equation}
\label{eq:syndrome}
H \blde_0^\top = H \blde_1^\top = \cdots = H \blde_L^\top
\end{equation}
and~(\ref{eq:strongly}) holds
for $\tau = L(n-k)/(L+1)$, namely
\begin{equation}
\label{eq:notlightly}
\sum_{m \in [0:L]} \weight(\blde_m) \le L(n-k) .
\end{equation}
By~(\ref{eq:ratelightly}) we have $L(n-k) \le n$,
so we can extend each support
$\Support(\blde_m)$ to a subset $J_m \subseteq [n]$ so that
the subsets $J_0, J_1, \ldots, J_L$ are disjoint
and satisfy conditions (S1)--(S2). From~(\ref{eq:syndrome})
it follows that the column vector
\[
\left(
(\blde_0)_{J_0}
\,|\, (\blde_1)_{J_1}
\,|\, \ldots
\,|\, (\blde_L)_{J_L}
\right)^\top ,
\]
which is nonzero and of length $L(n-k)$,
is in the right kernel of $M_{J_0,J_1,\ldots,J_L}(H)$,
namely, this matrix,
which is of order $(L(n-k)) \times (L(n-k))$, is singular.

Turning to the ``only if'' part,
suppose that
$M_{J_0,J_1,\ldots,J_L}(H)$ is singular for
disjoint nonempty subsets $J_0, J_1, \ldots, J_L$ that satisfy
conditions (S1)--(S2). Then its right kernel contains
a nonzero column vector,
\begin{equation}
\label{eq:rightkernellightly}
\left(
\hat{\blde}_0
\,|\, \hat{\blde}_1
\,|\, \ldots
\,|\, \hat{\blde}_L
\right)^\top ,
\end{equation}
where $\hat{\blde}_m \in F^{|J_m|}$.

For $m \in [0:L]$, let $\blde_m \in F^n$ be defined by
\[
(\blde_m)_{J_m} = \hat{\blde}_m
\quad \textrm{and} \quad
(\blde_m)_{[n] \setminus J_m} = \bldzero .
\]
It is easy to see that
the $L+1$ vectors $\blde_0, \blde_1, \ldots, \blde_L$
satisfy~(\ref{eq:syndrome}),
namely, they all belong to the same coset of~$\code$ within $F^n$.
These vectors have disjoint supports and,
by condition~(S1), they are all in $\Ball(n,n-k)$.
Now, if the coset they belong to were the trivial coset~$\code$
(namely, if they were all codewords of~$\code$),
then each $\blde_m \; (\in \Ball(n,n-k))$ would be zero,
but then the vector in~(\ref{eq:rightkernellightly}) would be all-zero.
We therefore conclude that each $\blde_m$ is nonzero.
Finally, condition~(S2) implies~(\ref{eq:notlightly}),
i.e., $\code$ is not lightly-$L$-MDS.
\end{proof}

\begin{remark}
\label{rem:S3}
The lemma holds also when the following condition is added:
\begin{list}{}{\settowidth{\labelwidth}{\textrm{(S2)}}}
\itemsep.5ex
\item[S3)]
$\displaystyle |J_m| + |J_\ell| \ge n-k+1$,
for all $0 \le m < \ell \le L$.
\end{list}
Referring to the ``if'' part of the proof,
for any two distinct vectors $\blde_m$ and $\blde_\ell$ therein
we have
$|J_m| + |J_\ell| \ge \weight(\blde_m) + \weight(\blde_\ell)
\ge \weight(\blde_m-\blde_\ell) \ge n-k+1$.\qed
\end{remark}
\fi

\begin{remark}
\label{rem:MDS,L=2}
For the special case $L = 2$,
condition~(S1) can be tightened to
\begin{list}{}{\settowidth{\labelwidth}{\textrm{(S2)}}}
\itemsep.5ex
\item[S1)]
$\displaystyle 2 \le |J_m| \le n-k-1$, for $m \in [0:2]$.
\end{list}
\ifPAGELIMIT
\else
Referring again to the ``if'' part of the proof,
the inequality~(\ref{eq:notlightly}) becomes in this case
\[
\sum_{m \in [0:2]} \weight(\blde_m) \le 2(n-k) .
\]
Combining this with $\weight(\blde_1) + \weight(\blde_2) \ge n-k+1$
then implies $\weight(\blde_0) \le n-k-1$.
By symmetry, the same applies to the weights of $\blde_1$ and $\blde_2$.
Hence, $\weight(\blde_1) \ge n-k+1 - \weight(\blde_2) \ge 2$
(and the same holds for $\weight(\blde_0)$ and $\weight(\blde_2$)).\qed
\fi
\end{remark}

\section{2-MDS codes}
\label{sec:2-MDS}

In this section, we consider the case of $2$-MDS codes over~$F$.
By Theorem~\ref{thm:stronglyonlyMDS} it follows that when $q > 2$,
such codes are necessarily
\ifPAGELIMIT
\else
MDS. The case $L = q = 2$ is covered by
the next lemma, which is proved in Appendix~\ref{sec:skipped}.

\begin{lemma}
\label{lem:L=q=2}
A linear $[n,k,d]$ code over $\GF(2)$ is $2$-MDS, if and only if
it is MDS, namely, 
$(k,d) \in \{ (1,n), (n-1,2), (n,1) \}$.
\end{lemma}
\fi

The next theorem provides a necessary and sufficient condition
for a linear $[n,k]$ MDS code to be $2$-MDS
(and, thus, $(\lfloor (2(n-k)/3 \rfloor,2)$-list decodable),
in terms of the light list decodability of
its punctured codes.\footnote{%
The puncturing of a code $\code \subseteq F^n$ on a subset
$X \subseteq [n]$ is the code
$\left\{ (\bldc)_{[n] \setminus X} \,:\, \bldc \in \code \right\}$.}
\ifPAGELIMIT
\else
For the special case of GRS codes
(to be discussed in Section~\ref{sec:GRS}),
the sufficiency part of the theorem is the dual-code version of
Lemma~4.4 in~\cite{ST};
we will make a general statement about duality in
Theorem~\ref{thm:dualityL=2} below.\footnote{%
The proof of Lemma~4.4 in~\cite{ST} makes use of
(the transpose of) the matrix~(\ref{eq:Malt}),
where~$H$ is taken as the \emph{generator matrix} of the GRS code.
In fact, the proof of Lemma~4.4 in~\cite{ST} inspired our
upcoming proof of Theorem~\ref{thm:dualityL=2} below.}
\fi

\begin{theorem}
\label{thm:puncturingL=2}
A linear $[n,k]$ MDS code $\code$ over~$F$
is $2$-MDS, if and only if for every integer~$w$ in the range
\begin{equation}
\label{eq:wrange}
\max \{ 0, n-2k \} \le w \le n-k-3 ,
\end{equation}
every linear $[n^*{=}n{-}w,k]$ code~$\code^*$ that is obtained by
puncturing~$\code$ on any~$w$ coordinates is lightly-$2$-MDS.
\end{theorem}

\begin{proof}
\ifPAGELIMIT
\else
The range~(\ref{eq:wrange}) is empty when $\min \{ k, n-k \} \le 2$,
yet we showed in Examples~\ref{ex:LMDSn-k=2}--\ref{ex:LMDSk=2,L=2}
that in this case, every linear $[n,k]$ MDS code is (unconditionally)
$2$-MDS. Therefore, we assume from now on in the proof that
$3 \le k \le n-3$.

We start with the ``if'' part.
\fi
Suppose that~$\code$ is MDS but not $2$-MDS.
Then there exist three distinct vectors
$\blde_0, \blde_1, \blde_2 \in F^n$
that belong to the same coset of~$\code$ such that
\begin{equation}
\ifPAGELIMIT
\else
\label{eq:sumsets}
\fi
\sum_{m \in [0:2]} \weight(\blde_m) \le 2(n-k) .
\end{equation}
By possibly translating these vectors by the same vector,
we can assume without loss of generality that
\begin{equation}
\label{eq:disjoint1}
\bigcap_{m \in [0:2]} \Support(\blde_m) = \emptyset .
\end{equation}
We will show that when we puncture~$\code$ on
the coordinates on which any \emph{two} of these supports intersect,
the resulting code is not lightly-$2$-MDS.

For $m \in [0:2]$, let $J_m \subseteq [n]$ be such that
$\Support(\blde_m) \subseteq J_m$
and the property~(\ref{eq:disjoint1}) extends
to $J_0$, $J_1$, and $J_2$, namely:
\begin{equation}
\ifPAGELIMIT
\else
\label{eq:disjoint2}
\fi
\bigcap_{m \in [0:2]} J_m = \emptyset
\end{equation}
(clearly, these conditions on $J_0$, $J_1$, and $J_2$ hold
if each $J_m$ is taken to be $\Support(\blde_m)$).
For $0 \le m < \ell \le 2$, let
\[
w_{m,\ell} = |J_m \cap J_\ell| .
\]
We have
\begin{eqnarray}
w_{m,\ell} & = & |J_m| + J_\ell| - |J_m \cup J_\ell|
\nonumber \\
& \le & |J_m| + |J_\ell|
- \left| \Support(\blde_m) \cup \Support(\blde_\ell) \right|
\nonumber \\
\label{eq:w01}
& \le & |J_m| + |J_\ell| - (n-k+1) ,
\end{eqnarray}
where the last step follows from the minimum distance of~$\code$.
Hereafter, we will further assume that
\begin{equation}
\label{eq:sumJ}
\sum_{m \in [0:2]} |J_m| = 2(n-k)
\end{equation}
by taking $J_m = \Support(\blde_m)$ for $m \in  \{0,1\}$ and
selecting $J_2$ to be of size $2(n-k)-|J_0|-|J_1|$ such that
\[
\Support(\blde_2) \subseteq J_2 \subseteq [n] \setminus (J_0 \cap J_1) .
\]
Indeed, by~(\ref{eq:disjoint1}), the set on the left is fully
contained in the set on the right which, in turn, has size
\[
n - w_{0,1}
\stackrel{\textrm{(\ref{eq:w01})}}{\ge} 2n-k+1-|J_0|-|J_1|
> 2(n-k)-|J_0|-|J_1| .
\]

Write
\ifPAGELIMIT
\else
\[
J = (J_0 \cap J_1) \cup (J_0 \cap J_2) \cup (J_1 \cap J_2) ,
\]
which is of size
\begin{eqnarray}
w = |J|
& \stackrel{\textrm{(\ref{eq:disjoint2})}}{=} &
w_{0,1} + w_{0,2} + w_{1,2}
\nonumber \\
& \stackrel{\textrm{(\ref{eq:w01})}}{\le} &
2 \sum_{m \in [0:2]} |J_m| - 3(n-k+1)
\nonumber \\
\label{eq:wsum}
& \stackrel{\textrm{(\ref{eq:sumJ})}}{=} & n-k-3 .
\end{eqnarray}
We find the allowable range of the size~$w$ of~$J$.
On the one hand, by~(\ref{eq:wsum}) we have $w \le n-k-3$.
On the other hand, by~(\ref{eq:disjoint2})
and the inclusion--exclusion principle we also have
\[
2(n-k) - w
\stackrel{\textrm{(\ref{eq:sumJ})}}{=}
\sum_{m \in [0:2]} |J_m| - |J|
= |J_0 \cup J_1 \cup J_2| \le n ,
\]
namely, $w \ge n-2k$.
Thus, $\max \{ 0, n-2k \} \le w \le n-k-3$
(as in~(\ref{eq:wrange})).

\fi
For $m \in [0:2]$, let $J^*_m = J_m \setminus J$.
\ifPAGELIMIT
\else
The sets $J_0^*$, $J_1^*$, and $J_2^*$
\fi
are disjoint and
\begin{equation}
\label{eq:sumJstar}
\sum_{m \in [0:2]} |J^*_m|
= \sum_{m \in [0:2]} |J_m| - 2|J|
\stackrel{\textrm{(\ref{eq:sumJ})}}{=}
2(n-k-w) .
\end{equation}
Moreover,
\begin{eqnarray}
|J^*_0|
& = &
\ifPAGELIMIT
\else
|J_0| - w_{0,1} - w_{0,2} \nonumber \\
& = & |J_0| + w_{1,2} - w \nonumber \\
\fi
& \stackrel{\textrm{(\ref{eq:w01})}}{\le} &
\sum_{m \in [0:2]} |J_m| - (n-k+1) - w \nonumber \\
\label{eq:sizeJstar}
& \stackrel{\textrm{(\ref{eq:sumJ})}}{=} &
n-k-w-1 ,
\end{eqnarray}
and the same upper bound applies to $|J^*_1|$ and $|J^*_2|$ as well.
Also note that~(\ref{eq:sumJstar}) and~(\ref{eq:sizeJstar})
imply for each $m \in [0:2]$ that
\begin{equation}
\label{eq:minJstar}
|J^*_m| \ge 2 .
\end{equation}

Let~$H$ be an $(n-k) \times n$ parity-check matrix of~$\code$
and let~$P$ be an $(n-k-w) \times (n-k)$ matrix whose rows form a basis
of the left kernel of $(H)_J$.
Write $J' = [n] \setminus J$
and let $H^* = (P H)_{J'}$ and $n^* = n-w$.
It can be readily verified that~$H^*$
is an $(n^*-k) \times n^*$ parity-check matrix of
the linear $[n^*,k]$ MDS code $\code^*$ which is obtained by
puncturing~$\code$ on the coordinate
\ifPAGELIMIT
\else
set~$J$.\footnote{%
This is easily seen if we first apply elementary linear operations
to the rows of~$H$ so that the~$w$ columns of $(H)_J$ become (distinct)
standard unit vectors in $F^{n-k}$. The rows of~$P$ can be taken as
the remaining elements of the (transposed) standard basis of $F^{n-k}$,
which means that $(P H)_{J'}$ is obtained from~$H$ simply by
removing the rows and columns that contain the $1$'s in $(H)_J$.}
\fi

For $m \in [0:2]$, define the vector $\blde^*_m \in F^{n^*}$ by
$\blde^*_m = (\blde_m)_{J'}$. From
\[
H (\blde_2 - \blde_0)^\top = H (\blde_1 - \blde_0)^\top = \bldzero
\]
we get that the vectors
$(H)_{J'} (\blde^*_m - \blde^*_0)^\top$
for $m \in [2]$ are in the linear span of
the columns of $(H)_J$; as such, these vectors are
in the right kernel of~$P$, namely,
\[
H^* (\blde^*_0)^\top = H^* (\blde^*_1)^\top = H^* (\blde^*_2)^\top .
\]
Noting that $\Support(\blde^*_m) \subseteq J^*_m$
(with equality when $m \in \{ 0, 1 \}$), we conclude that 
\[
\,\;
\weight(\blde^*_m) \le |J^*_m|
\stackrel{\textrm{(\ref{eq:sizeJstar})}}{\le} n^*-k-1 ,
\]
\[
\sum_{m \in [0:2]} \weight(\blde^*_m)
\le \sum_{m \in [0:2]} |J^*_m|
\stackrel{\textrm{(\ref{eq:sumJstar})}}{=} 2(n^*-k) ,
\]
and, for $0 \le m < \ell \le 2$,
\[
\Support(\blde^*_m) \cap \Support(\blde^*_\ell)
\subseteq J^*_m \cap J^*_\ell = \emptyset .
\]
Moreover, by~(\ref{eq:minJstar}) we have
$\weight(\blde^*_m) = |J^*_m| \ge 2$ for $m \in \{ 0, 1 \}$,
namely, $\blde^*_0$ and $\blde^*_1$ are nonzero
in $\Ball(n^*,n^*-k-1)$ and, therefore, are in a nontrivial
coset of~$\code^*$ within $F^{n^*}$. This, in turn, implies
that $\blde^*_2$ (which is in the same coset) is nonzero too.
Thus, starting off with the assumption
that~$\code$ is not $2$-MDS,
we have shown that~$\code^*$ is not lightly-$2$-MDS.
\ifPAGELIMIT
\else

Turning to the ``only if'' part, the proof is carried out by
essentially retracing our steps for the ``if'' part.
Let~$H$ be an $(n-k) \times n$ parity-check matrix of~$\code$ and
let a code~$\code^*$ be given that is obtained
by puncturing~$\code$ on the coordinates that are indexed by
some subset $J \subseteq [n]$ of size $|J| = w$.
Suppose that~$\code^*$ is not lightly-$2$-MDS, namely, there exist
nonzero vectors $\blde^*_0, \blde^*_1, \blde^*_2 \in \Ball(n-w,n-w-k)$
of disjoint supports such that
\begin{equation}
\label{eq:sumweighteprime}
\sum_{m \in [0:2]} \weight(\blde^*_m) \le 2(n-k-w)
\end{equation}
and
\begin{equation}
\label{eq:Hstar}
H^* (\blde^*_1 - \blde^*_0)^\top
= H^* (\blde^*_2 - \blde^*_0)^\top = \bldzero ,
\end{equation}
where $H^* = (P H)_{J'}$
is an $(n-k-w) \times (n-w)$ parity-check matrix of~$\code^*$,
with~$P$ being an $(n-k-w) \times (n-k)$ matrix whose rows form
a basis of the left kernel of $(H)_J$. From~(\ref{eq:Hstar})
it follows that the vectors
$(H)_{J'} (\blde^*_m - \blde^*_0)^\top$
for $m \in [2]$ are in the right kernel of~$P$,
which means that they are in the linear span of the columns of $(H)_J$;
namely, there exist $\bldx_1, \bldx_2 \in F^w$ such that,
for $m \in [2]$:
\begin{equation}
\label{eq:e}
(H)_{J'} (\blde^*_m - \blde^*_0)^\top
+ (H)_J \bldx_m^\top = \bldzero .
\end{equation}

For $m \in [0:2]$, define the vector $\blde_m \in F^n$ by
$(\blde_m)_{J'} = \blde^*_m$ and
$(\blde_m)_J = \bldx_m$, where $\bldx_0 = \bldzero$.
By~(\ref{eq:sumweighteprime}) and~(\ref{eq:e}) we get, respectively,
that
\begin{eqnarray*}
\sum_{m \in [0:2]} \weight(\blde_m)
& = &
\sum_{m \in [0:2]} \weight(\blde^*_m)
+ \sum_{m \in [0:2]} \weight(\bldx_m) \\
& \le &
\sum_{m \in [0:2]} \weight(\blde^*_m) + 2 w \le 2(n-k)
\end{eqnarray*}
and
\[
H \blde_0^\top = H \blde_1^\top = H \blde_2^\top .
\]
Moreover, the vectors $\blde_m$ are distinct
since the vectors $\blde^*_m$ are nonzero with disjoint supports.
Hence, $\code$ is not $2$-MDS.
\fi
\end{proof}

\ifPAGELIMIT
\else
\begin{remark}
\label{rem:lightly}
When applying Theorem~\ref{thm:puncturingL=2}
to test whether a given linear $[n,k]$ MDS code~$\code$ is $2$-MDS,
we can use Lemma~\ref{lem:lightly} to check if
each punctured $[n{-}w,k]$ code~$\code^*$ is lightly-$2$-MDS.
Specifically, since~$\code$ is MDS then so is~$\code^*$;
moreover, for the range~(\ref{eq:wrange}),
the rate of~$\code^*$ is $k/(n-w) \ge 1/2$, namely,
the inequality~(\ref{eq:ratelightly}) holds as well.\qed
\end{remark}
\fi

It is well known that the MDS property is preserved under
duality~\cite[p.~318]{MS}.
\ifPAGELIMIT
\else
We next apply
\fi
Theorem~\ref{thm:puncturingL=2}
to show that the same can be said about the $2$-MDS property.

\begin{theorem}
\label{thm:dualityL=2}
A linear $[n,k]$ code over~$F$ is $2$-MDS,
if and only if its dual code is.
\end{theorem}

\ifPAGELIMIT
\else
\begin{proof}
Suppose that~$\code$ is MDS but not $2$-MDS,
namely, there exist three distinct vectors
$\blde_0, \blde_1, \blde_2 \in F^n$
that belong to the same coset of~$\code$ such that~(\ref{eq:sumsets})
holds. We proceed by applying the proof of
Theorem~\ref{thm:puncturingL=2}---verbatim---up to
Eq.~(\ref{eq:minJstar}).
Henceforth, we show that when we puncture 
the dual code~$\code^\perp$ on the coordinates that belong
to \emph{none} of the supports of these three vectors,
the resulting code is not lightly-$2$-MDS.
This, in turn, will imply (by Theorem~\ref{thm:puncturingL=2})
that $\code^\perp$ is not $2$-MDS.

Let~$G$ be a $k \times n$ generator matrix of~$\code$
and let $\bldu_0$, $\bldu_1$, and $\bldu_2$ be distinct vectors in $F^k$
such that
\[
\blde_0 + \bldu_0 G = \blde_1 + \bldu_1 G = \blde_2 + \bldu_2 G .
\]
Letting
\begin{equation}
\label{eq:K}
K = \bigcup_{m \in [0:2]} J_m = J \cup \bigcup_{m \in [0:2]} J^*_m
\end{equation}
and denoting $K' = [n] \setminus K$, we have
\[
(\blde_1)_{J^*_0 \cup K'} = (\blde_2)_{J^*_0 \cup K'} = \bldzero
\]
and, so,
\[
(\bldu_1 - \bldu_2) (G)_{J^*_0 \cup K'} = \bldzero .
\]
Similarly,
\[
(\bldu_2 - \bldu_0) (G)_{J^*_1 \cup K'} = \bldzero
\quad \textrm{and} \quad
(\bldu_0 - \bldu_1) (G)_{J^*_2 \cup K'} = \bldzero .
\]
Defining $\blda_m = \bldu_{m+1} - \bldu_{m+2}$ (with indexes
taken modulo~$3$) we thus get for every $m \in [0:2]$
that $\blda_m \ne \bldzero$ and
\begin{equation}
\label{eq:am1}
\blda_m (G)_{J^*_m \cup K'} = \bldzero .
\end{equation}
Moreover,
\begin{equation}
\label{eq:am2}
\sum_{m \in [0:2]} \blda_m = \bldzero .
\end{equation}

Write
\begin{equation}
\label{eq:wstar}
w^* = |K'| = n - |K|
\stackrel{{\textrm{(\ref{eq:sumJstar})}}+{\textrm{(\ref{eq:K})}}}{=}
2k-n+w ,
\end{equation}
let~$P$ be a $(k-w^*) \times k$ matrix whose rows form
a basis of the left kernel of $(G)_{K'}$,
and let $H^* = (P G)_K$.
We observe that~$H^*$ is a $(k-w^*) \times (n-w^*)$
parity-check matrix of
the linear $[n^*{=}n{-}w^*,k^*{=}n{-}k]$ code, $\code^*$,
which is obtained
by puncturing the dual code $\code^\perp$ on the coordinate set~$K'$.
It follows from~(\ref{eq:am1}) that each vector $\blda_m$ belongs to
the row span of~$P$, namely, we can write $\blda_m = \bldb_m P$
for a unique nonzero $\bldb_m \in F^{k-w^*}$.
By~(\ref{eq:am1})--(\ref{eq:am2}) we conclude that for $m \in [0:2]$,
\[
\bldb_m (H^*)_{J^*_m} = \bldzero ,
\]
and, in addition,
\[
\sum_{m \in [0:2]} \bldb_m = \bldzero ,
\]
namely, the (nonzero) vector $\left( \bldb_1 \,|\, \bldb_2 \right)$
is in the left kernel of the matrix $M_{J^*_0,J^*_1,J^*_2}(H^*)$
(as in~(\ref{eq:M})).
By~(\ref{eq:sumJstar}) and~(\ref{eq:wstar}),
this matrix has $2(k-w^*) = 2(n-k-w)$ rows and
the same number of columns and, therefore, it is singular.

Finally, we show that~$\code^*$ and the subsets
$J^*_0$, $J^*_1$, and $J^*_2$ satisfy the conditions
of Lemma~\ref{lem:lightly} for list size $L = 2$
and code parameters $[n^*,k^*]$.
These subsets are obviously disjoint;
by~(\ref{eq:minJstar}) they are nonempty;
the code rate is given by
\[
\frac{k^*}{n^*} = \frac{n-k}{n-w^*}
\stackrel{\textrm{(\ref{eq:wstar})}}{=}
\frac{n-k}{2(n-k)-w}
\ge \frac{1}{2}
\]
(as in~(\ref{eq:ratelightly}));
$\code^*$ is MDS;
and~(\ref{eq:sizeJstar}) and~(\ref{eq:sumJstar})
(along with the equality $n^*-k^* = n-k-w$)
imply, respectively, conditions~(S1) and~(S2).

Thus, starting off with the assumption
that~$\code$ is not $2$-MDS,
we obtain from Lemma~\ref{lem:lightly}
that~$\code^*$ is not lightly-$2$-MDS.
Hence, by Theorem~\ref{thm:puncturingL=2} we conclude
that $\code^\perp$ is not $2$-MDS.
\end{proof}
\fi

We note that Theorem~\ref{thm:dualityL=2} does not generalize to
$3$-MDS
\ifPAGELIMIT
\else
codes: the code in Example~\ref{ex:LMDSk=2,L=3} is not $3$-MDS
while, by Example~\ref{ex:LMDSn-k=2}, its dual code is.
\fi

Another application of Theorem~\ref{thm:puncturingL=2}
is the next result,
which states that over sufficiently large fields
(namely, exponential in the code length),
almost all linear codes are $2$-MDS.
The proof can be found
\ifPAGELIMIT
\else
in Appendix~\ref{sec:skipped}.
\fi

\begin{theorem}
\label{thm:almostall2-MDS}
Given $n \in \Integers^+$ and $k \in [n]$,
all but a fraction $O \left( 5^n/q \right)$ of
the linear $[n,k]$ codes over~$F$ are $2$-MDS.
\end{theorem}

\section{The $2$-MDS GRS case}
\label{sec:GRS}

In this section, we describe an explicit construction of $2$-MDS
$[n,n{-}\rho]$ GRS codes over extension fields~$F$ of $\GF(2)$
of size $q$ which is polynomial~$n$, provided that
the redundancy~$\rho$ is regarded as a constant
(by Theorem~\ref{thm:dualityL=2}, the dual code
will be a $2$-MDS $[n,\rho]$ code over~$F$). The degree of
the polynomial in~$n$, however, grows rapidly with~$\rho$,
so this result may have a limited practical value;
still, the field size herein is generally much smaller than that in
the explicit construction of~\cite{ST}.
\ifPAGELIMIT
\else
The construction will be presented in
Section~\ref{sec:GRSgeneral};
then, in Section~\ref{sec:GRSn-k=3},
we fine-tune the construction to yield more favorable parameters
for the special case $\rho = 3$.
Section~\ref{sec:GRSnotation} introduces some notation
that will be used in the analysis of the construction.
\fi

\subsection{Notation and preliminary analysis}
\label{sec:GRSnotation}

Recall that an $[n,k]$ GRS code over~$F$ is
a linear $[n,k]$ code with an $(n-k) \times n$ parity-check matrix
$H_\GRS = (H_{i,j})$ of the form
\begin{equation}
\label{eq:HGRS}
H_{i,j} = v_j \alpha_j^i ,
\quad
i \in [0:n{-}k{-}1] , \; j \in [n] ,
\end{equation}
where $\alpha_1, \alpha_2, \ldots, \alpha_n$ are
the code locators, which are distinct elements of~$F$,
and $v_1, v_2, \ldots, v_n$ are the column multipliers,
which are nonzero elements of~$F$.
\ifPAGELIMIT
\else
GRS codes are MDS and are closed under puncturing
and under duality~\cite[pp.~303--304]{MS}.
While the freedom of selecting the column multipliers is needed
in order to establish these closures, the choice of their values is
immaterial for the purpose of this work, so we will assume henceforth
that they
\fi
are all~$1$. Substituting $v_j = 1$ in~(\ref{eq:HGRS})
and writing $\bldalpha = (\alpha_j)_{j \in [n]}$,
the respective GRS code will be denoted by $\code_k(\bldalpha)$.

Given an integer $\rho \ge 3$,
let $\bldrho = (\rho_0,\rho_1,\rho_2)$ be a partition of $2\rho$
where $2 \le \rho_0 \le \rho_1 \le \rho_2 < \rho$
are positive integers (that sum to $2\rho$)
and let $(\Upsilon_0,\Upsilon_1,\Upsilon_2)$
be the following partition of the set $[2\rho]$:
\ifPAGELIMIT
\else
\begin{eqnarray}
\Upsilon_0 & = & [\rho_0] , \nonumber \\
\label{eq:Upsilon}
\Upsilon_1 & = & [\rho_0+1:\rho_0+\rho_1] , \\
\Upsilon_2 & = & [\rho_0+\rho_1+1:2\rho]
\nonumber
\end{eqnarray}
(so that $|\Upsilon_m| = \rho_m$, for $m \in [0:2]$).
\fi
For a vector $\bldx = (x_\ell)_{\ell \in [2\rho]}$
of indeterminates, we define the $(2\rho) \times (2\rho)$
parametrized matrix $M_\bldrho(\bldx)$ by
\begin{equation}
\ifPAGELIMIT
\else
\label{eq:M(x)}
\fi
M_\bldrho(\bldx) =
\left(
\arraycolsep0.15ex
\renewcommand{\arraystretch}{2.0}
\begin{array}{c|c|c}
(-x_\ell^i)_{i=0, \ell \in \Upsilon_0}^{\rho-1} & 
(x_\ell^i)_{i=0, \ell \in \Upsilon_1}^{\rho-1} &  \\
\hline
(-x_\ell^i)_{i=0, \ell \in \Upsilon_0}^{\rho-1} &  &
(x_\ell^i)_{i=0, \ell \in \Upsilon_2}^{\rho-1} 
\end{array}
\right) .
\end{equation}
Thus, $\det(M_\bldrho(\bldx)$ is
an element of the ring, $F[\bldx]$,
of multivariate polynomials in the entries of~$\bldx$ over~$F$.

When applying Lemma~\ref{lem:lightly}
to test whether a given $[n,k]$ GRS code $\code_k(\bldalpha)$
is lightly-$2$-MDS, we need to check whether
\ifPAGELIMIT
\else
$\det(M_{J_0,J_1,J_2}(H_\GRS)) \ne 0$
for all triples $(J_0,J_1,J_2)$
of disjoint nonempty subsets of $[n]$ that satisfy
condition~(S1) in Remark~\ref{rem:MDS,L=2}
and condition~(S2) in Lemma~\ref{lem:lightly}.
Equivalently,we need to check whether
\[
\det(M_\bldrho(\bldx)) \ne 0
\]
\fi
for all partitions $\bldrho = (\rho_0,\rho_1,\rho_2)$ of $2(n-k)$ with
$2 \le \rho_0 \le \rho_1 \le \rho_2 < n-k$
and for all triples $(J_0,J_1,J_2)$
of disjoint subsets of $[n]$ of sizes $|J_m| = \rho_m$,
while substituting
$(\bldx)_{\Upsilon_m} = (\bldalpha)_{J_m}$, for $m \in [0:2]$.

\ifPAGELIMIT
\else
\begin{example}
\label{ex:GRSn-k=3}
For $\rho = 3$ and the partition $(2,2,2)$:
\[
M_{2,2,2}(\bldx) =
\left(
\arraycolsep1ex
\begin{array}{cc|cc|cc}
-1     & -1     & 1     & 1     &       &       \\
-x_1   & -x_2   & x_3   & x_4   &       &       \\
-x_1^2 & -x_2^2 & x_3^2 & x_4^2 &       &       \\
\hline
-1     & -1     &       &       & 1     & 1     \\
-x_1   & -x_2   &       &       & x_5   & x_6   \\
-x_1^2 & -x_2^2 &       &       & x_5^2 & x_6^2
\end{array}
\right)
,
\]
and we have
\begin{eqnarray*}
\det(M_{2,2,2}(\bldx))
& = &
-(x_2 - x_1)(x_4 - x_3)(x_6 - x_5) \\
&& \quad\quad {} \cdot \bigl(
x_1 x_2 (x_3 + x_4 - x_5 - x_6) \\
&& \quad\quad {} + x_3 x_4 (x_5 + x_6 - x_1 - x_2) \\
&& \quad\quad {} + x_5 x_6 (x_1 + x_2 - x_3 - x_4) \bigr) \\
& = &
-(x_2 - x_1)(x_4 - x_3)(x_6 - x_5) \\
&& \quad\quad {} \cdot \det(S_{2,2,2}(\bldx)) ,
\end{eqnarray*}
where
\begin{equation}
\label{eq:S}
S_{2,2,2}(\bldx) =
\left(
\begin{array}{ccc}
1   & -x_1 - x_2 & x_1 x_2 \\
1   & -x_3 - x_4 & x_3 x_4 \\
1   & -x_5 - x_6 & x_5 x_6
\end{array}
\right)
.
\end{equation}
Namely, the expansion of $\det(S_{2,2,2}(\bldx))$
yields the sum of all~$12$ monomials of the form
\begin{equation}
\label{eq:GRSn-k=3}
\pm x_j \cdot \prod_{\ell \in \Upsilon_m} x_\ell ,
\end{equation}
where $m \in [0:2]$ and $j \in \Upsilon_{m+1} \cup \Upsilon_{m+2}$
(with the indexes $m+1$ and $m+2$ taken
modulo~$3$),
and the minus sign is taken when $j \in \Upsilon_{m+2}$.

For $n-k = 3$, the subsets $J_m$ 
in the test of Lemma~\ref{lem:lightly}
(with condition~(S1) taken from Remark~\ref{rem:MDS,L=2})
are all of size~$2$,
namely, $\bldx$ will range over all (unordered) triples of
(unordered) pairs of code locators.\qed
\end{example}
\fi

For general $\rho \ge 3$ and partition $\bldrho$ of $[2\rho]$,
the Leibniz expansion of $\det(M_\bldrho(\bldx))$
results in a sum of monomials of the form
\begin{equation}
\label{eq:monomial}
\pm \prod_{\ell \in [2\rho]} x_\ell^{r_\ell} ,
\end{equation}
where $\bldr = (r_\ell)_{\ell \in [2\rho]}$
ranges over the elements in $[0:\rho-1]^{2\rho}$
that satisfy the following conditions:
\begin{list}{}{\settowidth{\labelwidth}{\textrm{(R2)}}}
\itemsep.5ex
\item[R1)]
for each
\ifPAGELIMIT
\else
    element
\fi
$r \in [0:\rho-1]$
there are exactly two distinct indexes~$\ell$ and~$\ell'$
in $[2\rho]$ for which $r_\ell = r_{\ell'} = r$, and---
\item[R2)]
those two indexes belong to distinct sets $\Upsilon_m$
\end{list}
(see~\cite[Lemma~4.2]{ST}).
The number of such monomials is
\ifPAGELIMIT
\else
given by\footnote{%
We do the enumeration over~$\bldr$
by selecting for each $m \in [0:2]$
some ordering on $\Upsilon_m$ and selecting a list
$\bldm = (m_0,m_1,\ldots,m_{\rho-1}) \in [0:2]^\rho$
such that each element $m \in [0:2]$ appears in~$\bldm$
exactly $\rho - \rho_m$ times.
Then, for $r = 0, 1, \ldots, \rho-1$, we set iteratively
$r_\ell = r_{\ell'} = r$, where~$\ell$ and~$\ell'$
are the next indexes in line
in~$\Upsilon_{m_r+1}$ and~$\Upsilon_{m_r+2}$, respectively,
according to the ordering that was selected on these subsets.
The formula for $N_\bldrho$ is
the product of the number of possible lists~$\bldm$
and the number of different orderings on each subset $\Upsilon_m$.}
\fi
\begin{eqnarray*}
\ifPAGELIMIT
\else
N_\bldrho
& = &
\binom{\rho}{\rho{-}\rho_0 \;\; \rho{-}\rho_1 \;\; \rho{-}\rho_2}
\prod_{m \in [0:2]} (\rho_m!) \\
\fi
& = &
\rho! \prod_{m \in [0:2]} \frac{\rho_m!}{(\rho-\rho_m)!}
\end{eqnarray*}
(which is always an even integer),
and each monomial is of total degree
\begin{equation}
\label{eq:degree}
\sum_{\ell \in [2\rho]} r_\ell
=
2 \sum_{r \in [0:\rho-1]} r = \rho(\rho-1) .
\end{equation}
\ifPAGELIMIT
\else
In particular, $\det(M_\bldrho(\bldx))$ is
a homogeneous multivariate polynomial in
the entries of~$\bldx$ over~$F$.

More properties of $\det(M_\bldrho(\bldx))$
are presented in Appendix~\ref{sec:M}.
\fi

Write
\[
N(\rho) =
\frac{1}{2} \max_\bldrho N_\bldrho ,
\]
where the maximization is over all partitions
$\bldrho = (\rho_0,\rho_1,\rho_2)$ of $2\rho$ 
such that $2 \le \rho_0 \le \rho_1 \le \rho_2 < \rho$.
It can be readily verified that
\ifPAGELIMIT
\else
the maximum is attained
when $\rho_0$, $\rho_1$, and $\rho_2$ are (as close as
they can get to being) equal, in which case the Stirling approximation
for the binomial coefficients yields
\[
N(\rho)
\approx
2 \sqrt{\pi \rho}
\cdot \left( \frac{4}{3 \, {\mathrm{e}}^2} \right)^\rho
\cdot \rho^{2 \rho}
\]
(where $\mathrm{e}$ is
the base of natural logarithms)~\cite[p.~309]{MS};
in particular,
\fi
\begin{equation}
\label{eq:Nrho}
N(\rho) < \rho^{2 \rho} .
\end{equation}

\ifPAGELIMIT
\else
\subsection{2-MDS GRS construction for small fixed redundancy}
\label{sec:GRSgeneral}

We describe in this section a construction
of a $2$-MDS $[n,n{-}\rho]$ GRS code
over extension fields~$F$ of $\GF(2)$
of size (much) smaller than $n^{\rho^{2\rho}}$.
\fi

Let $\rho \ge 3$ be an integer,
let $h \in \Integers^+$ be such that $2^{\rho(\rho-1)h} \ge N(\rho)$,
write $n = 2^h$ (which will be the code length),
and let~$K$ be the field $\GF(2^{\rho(\rho-1)h})$
(which is of size at least $N(\rho)$).
Also, let $\{ b_j \}_{j \in [n]}$ be the set of elements of $\GF(2^h)$
and let~$\beta$ be an element in~$K$
that is not in any proper subfield of $K$.
For $j \in [n]$, define the following~$n$ elements of~$K$:
\begin{equation}
\ifPAGELIMIT
\else
\label{eq:betaj}
\fi
\beta_j = \beta + b_j .
\end{equation}
These elements have the following property:
for any two multisets $J, J' \subseteq [n]$,
each of size${} \le \rho(\rho-1)$, the equality
\[
\prod_{j \in J} \beta_j = \prod_{j' \in J'} \beta_{j'}
\]
can hold only if $J = J'$
(see~\cite{Bose}, \cite{BoseChowla}, \cite{Chowla}).
As such, the set $\{ \beta_j \}_{j \in [n]}$ forms
a (generalized multiplicative) \emph{Sidon set}~\cite{BabaiSos}.

Fix arbitrarily $N(\rho)$ distinct elements
$\xi_1, \xi_2, \ldots, \xi_{N(\rho)} \in K$
and for each $j \in [n]$,
let $\lambda_j(z)$ be the (unique) polynomial of
degree${} < N(\rho)$ over~$K$ that interpolates through
the $N(\rho)$ points
$\left\{ (\xi_i,\beta_j^{2i-1}) \right\}_{i \in [N(\rho)]}$:
\begin{equation}
\label{eq:interpolate}
\lambda_j(\xi_i) = \beta_j^{2i-1} , \quad i \in [N(\rho)] .
\end{equation}

We now define the underlying field of the code to be
the extension field $F = \GF(2^{\mu(\rho)\rho(\rho-1)h})$
of~$K$ of extension degree
\[
\mu(\rho) = \rho(\rho-1) (N(\rho)-1) + 1 .
\]
Finally, letting~$\gamma$ be a root in~$F$ of
a degree-$\mu(\rho)$ irreducible polynomial
\ifPAGELIMIT
\else
over~$K$,\footnote{%
Since $\gcd(\mu(\rho),\rho(\rho-1)) = 1$,
we can take the polynomial to be
irreducible over $\GF(2^h)$~\cite[p.~107]{LN}.
Moreover, when $\gcd(\mu(\rho),h) = 1$
(which happens when, say, all the prime factors of~$h$ are prime factors
of $\rho(\rho-1)(N(\rho)-1)$, e.g., when $h$ is a power of~$2$),
we can take the polynomial to be irreducible over $\GF(2)$.}
\fi
our construction is
the GRS code $\code_{n-\rho}(\bldalpha)$ over~$F$
with the code locators
\begin{equation}
\label{eq:GRSlocators}
\alpha_j = \lambda_j(\gamma) , \quad j \in [n] .
\end{equation}

We turn to analyzing the construction.
Let $\kappa : [2\rho] \rightarrow [n]$ be an arbitrary injective mapping
and substitute
\begin{equation}
\label{eq:substitution}
\bldx = (x_\ell)_{\ell \in [2\rho]}
\leftarrow \left( \alpha_{\kappa(\ell)} \right)_{\ell \in [2\rho]}
\end{equation}
into the monomials~(\ref{eq:monomial}) to obtain
the following $N_\bldrho$ elements of~$F$:
\begin{equation}
\label{eq:term}
\prod_{\ell \in [2\rho]} \alpha_{\kappa(\ell)}^{r_\ell} ,
\end{equation}
where $\bldr = (r_\ell)_{\ell \in [2\rho]}$
ranges over all $(2\rho)$-tuples that satisfy conditions (R1)--(R2).
We show that the elements~(\ref{eq:term}) do not sum to zero, namely,
$\det(M_\bldrho(\bldx)) \ne 0$
under the substitution~(\ref{eq:substitution}).

Suppose to the contrary that
\ifPAGELIMIT
\else
the elements~(\ref{eq:term}) sum to zero, namely,
\fi
\begin{equation}
\ifPAGELIMIT
\else
\label{eq:sumzero}
\fi
\sum_\bldr
\prod_{\ell \in [2\rho]} \alpha_{\kappa(\ell)}^{r_\ell} = 0 .
\end{equation}
By the definition of the code locators in~(\ref{eq:GRSlocators}):
\begin{equation}
\label{eq:basis1}
\sum_\bldr
\prod_{\ell \in [2\rho]}
\left( \lambda_{\kappa(\ell)}(\gamma) \right)^{r_\ell} 
= 0 .
\end{equation}
On the other hand,
by the definition of the polynomials $\lambda_j(z)$:
\begin{equation}
\label{eq:basis2}
\deg
\prod_{\ell \in [2\rho]}
\left( \lambda_{\kappa(\ell)}(z) \right)^{r_\ell}
\stackrel{\textrm{(\ref{eq:degree})}}{\le} \rho(\rho-1)(N(\rho)-1) .
\end{equation}
Noting that $\{ \gamma^i \}_{i \in [0:\mu(\rho)-1]}$
is a basis of~$F$ over~$K$,
we get from~(\ref{eq:basis1}) and~(\ref{eq:basis2})
the following polynomial identity:
\begin{equation}
\ifPAGELIMIT
\else
\label{eq:basis3}
\fi
\sum_\bldr
\prod_{\ell \in [2\rho]}
\left( \lambda_{\kappa(\ell)}(z) \right)^{r_\ell} = 0 .
\end{equation}
Hence,
\[
\sum_\bldr
\prod_{\ell \in [2\rho]}
\left( \lambda_{\kappa(\ell)}(\xi_i) \right)^{r_\ell} = 0 ,
\quad i \in [N(\rho)] ,
\]
which, by~(\ref{eq:interpolate}), becomes
\ifPAGELIMIT
\else
\[
\sum_\bldr
\Bigl( \prod_{\ell \in [2\rho]}
\beta_{\kappa(\ell)}^{r_\ell}
\Bigr)^{2i-1} = 0 ,
\quad i \in [N(\rho)] ,
\]
or
\fi
\begin{equation}
\label{eq:basis4}
\sum_\bldr \theta_\bldr^{2i-1} = 0 , \quad i \in [N(\rho)] ,
\end{equation}
where
\[
\theta_\bldr = \prod_{\ell \in [2\rho]} \beta_{\kappa(\ell)}^{r_\ell} .
\]
Now, the elements $\theta_\bldr$ are distinct for distinct~$\bldr$
due to the Sidon property of the elements $\beta_j$,
which means that the respective column vectors,
\[
\left(
\theta_\bldr \;
\theta^3_\bldr \;
\ldots \;
\theta^{2N(\rho)-1}_\bldr
\right)^\top
\;
\left( \in K^{N(\rho)} \right) ,
\]
are linearly independent over $\GF(2)$, as they are
(at most $2 N(\rho)$) distinct columns of a parity-check matrix
of a binary BCH code
with minimum distance${} > 2 N(\rho)$~\cite[Ch.~7, \S 6]{MS}.
On the other hand, (\ref{eq:basis4}) implies that these vectors
sum to zero, thereby reaching a contradiction.

Thus, we have shown that $\det(M_\bldrho(\bldx)) \ne 0$
under the substitution~(\ref{eq:substitution})
and, so, we conclude from Lemma~\ref{lem:lightly} that
the code $\code_{n-\rho}(\bldalpha)$ is lightly-$2$-MDS.
Furthermore, going through the analysis
it is fairly easy to see that it implies
that any puncturing of this code is lightly-$2$-MDS too.
Hence, by Theorem~\ref{thm:puncturingL=2},
the code $\code_{n-\rho}(\bldalpha)$ is $2$-MDS.

The field size~$q$ of~$F$ is related to the code length~$n$
and the redundancy~$\rho$ by:
\[
q = 2^{\mu(\rho)\rho(\rho-1)h}
= n^{\mu(\rho)\rho(\rho-1)}
< n^{N(\rho)\rho^2(\rho-1)^2}
\stackrel{\textrm{(\ref{eq:Nrho})}}{\ll} n^{\rho^{2\rho+4}} .
\]
In comparison, (the dual code of)
the explicit construction in Theorem~1.7 in~\cite{ST}
requires a field size of $q = 2^{\rho^n}$.

\ifPAGELIMIT
\else
\subsection{Fine tuning for the case of redundancy 3}
\label{sec:GRSn-k=3}

We consider in this section
the special case of $\rho = 3$ (as in Example~\ref{ex:GRSn-k=3}).
We recall that by Theorem~\ref{thm:SpherePacking},
the field size in this case must be at least
$\Omega(n^2)$. In comparison, Theorem~1.6 in~\cite{ST}
implies the existence of (the dual code of)
a $2$-MDS $[n,n{-}3]$ GRS code over a field of size $O(n^6)$.
The argument therein in fact suggests that such a code
can be constructed by a randomized algorithm requiring
$O(n^6)$ operations in~$F$.\footnote{%
In fact, it is fairly easy to show
that when $|F|$ grows with~$n$ at least as $n^6$ then,
with high probability, any $3 \times n$ matrix over~$F$
is a parity-check matrix of a linear MDS code over~$F$ which is also
lightly-$2$-MDS and, thus, by Theorem~\ref{thm:puncturingL=2},
it is $2$-MDS.}

Using the following (deterministic) iterative procedure,
the field~$F$ can be taken to be of size $O(n^5)$.
We select the first five code locators $(\alpha_j)_{j \in [5]}$
to be arbitrary distinct elements of~$F$.
Assuming now that we have selected the code locators
$(\alpha_j)_{j \in [t-1]}$ for some $t > 5$,
we select $\alpha_t$ to be a new element in~$F$ such that
$\det(S_{2,2,2}(\bldx)) \ne 0$
(see~(\ref{eq:S})),
where we substitute
$(\bldx)_{\Upsilon_m}
\leftarrow (\alpha_1 \, \alpha_2 \, \ldots \, \alpha_t)_{J_m}$,
$m \in [0:2]$, with $\{ J_0, J_1, J_2 \}$
ranging over all the (unordered) triples of disjoint subsets of $[t]$
such that $t \in J_0$.
In other words, $\alpha_t$ should \emph{not} solve
the following linear equation in $x_1$,
\begin{eqnarray*}
\lefteqn{
\bigl( x_2(x_3 + x_4 - x_5 - x_6) - x_3 x_4 + x_5 x_6 \bigr) \cdot x_1 
} \makebox[5ex]{} \\
& = &
x_3 x_4 (x_2 - x_5 - x_6) - x_5 x_6 (x_2 - x_3 - x_4) ,
\end{eqnarray*}
for any assignment of (already selected) distinct code locators to
$x_2, x_3, \ldots, x_6$. For any such assignment, this equation
has at most one solution, unless the coefficients on both sides
vanish, namely:
\begin{eqnarray*}
x_2(x_3 + x_4 - x_5 - x_6) - x_3 x_4 + x_5 x_6 & = & 0 \\
x_3 x_4 (x_2 - x_5 - x_6) - x_5 x_6 (x_2 - x_3 - x_4) & = & 0 .
\end{eqnarray*}
Yet these two equations cannot hold simultaneously. Indeed,
regarding them as linear equations in $x_2$,
they could both hold only if the $2 \times 2$ matrix
\[
\left(
\arraycolsep1ex
\begin{array}{cc}
x_3 + x_4 - x_5 - x_6 & - x_3 x_4 + x_5 x_6 \\
x_3 x_4 - x_5 x_6 &
-x_3 x_4 (x_5 + x_6) + x_5 x_6 (x_3 + x_4)
\end{array}
\right)
\]
were singular; but the determinant of this matrix equals
$(x_3 - x_5)(x_3 - x_6)(x_4 - x_5)(x_4 - x_6)$,
which is nonzero for all the assignments to
$x_3, x_4, x_5, x_6$. It readily follows that when
\[
15 \cdot \binom{n-1}{5} < q ,
\]
we will be able to find a qualifying value for $\alpha_t$,
as long as $t \le n$.

In what follows,
we modify the construction of Section~\ref{sec:GRSgeneral}
to yield an explicit construction of 
a $2$-MDS $[n,n{-}3]$ GRS code over $\GF(n^{32})$,
where~$n$ is an odd power of~$2$.
Here, by ``explicit'' we mean that for any $j \in [n]$,
the complexity of computing the $j$th code locator
is polylogarithmic
in~$n$.\footnote{%
Complexity is measured here by binary operations,
as the arithmetic of the fields involved---including finding their
representations---can be carried out
in time complexity which is polylogarithmic in
the field size~\cite{Shoup}.}

Let $h \ge 3$ be an odd positive integer,
write $n = 2^h$, and let~$\beta$ be
an element in $K = \GF(2^{2h})$ that is not
in $\GF(2^h)$.\footnote{%
In the special case where~$h$ is a power of~$3$,
the field~$K$ can be constructed as the polynomial ring
modulo the binary polynomial $x^{2h} + x^h + 1$,
which is irreducible over~$\GF(2)$~\cite[p.~96]{Golomb}.}
For $j \in [n]$, define the elements $\beta_j \in K$
as in~(\ref{eq:betaj}):
here the Sidon property means that
$\beta_i \beta_j = \beta_{i'} \beta_{j'}$
implies $\{ i, j \} = \{ i', j' \}$.
Fix six distinct elements $\xi_1, \xi_2, \ldots, \xi_6 \in K$
and for each $j \in [n]$,
let $\lambda_j(z)$ be of degree${} < 6$ over~$K$
that interpolates through
$\left\{ (\xi_i,\beta_j^{2i-1}) \right\}_{i \in [6]}$
as in~(\ref{eq:interpolate}).

Let $F = \GF(2^{32h}) = \GF(n^{32})$ and let~$\gamma$
be a root in~$F$ of a degree-$16$ irreducible polynomial
over $\GF(4)$ (e.g., the polynomial $x^{16} + x^3 + x^2 + \omega$,
where~$\omega$ is a root of $x^2 + x + 1$);
since~$h$ is odd, this polynomial is also irreducible over~$K$.
Our construction is the GRS code $\code_{n-3}(\bldalpha)$ with
the code locators as in~(\ref{eq:GRSlocators}).

Let $\kappa : [6] \rightarrow [n]$ be an injective mapping
and let $\Upsilon_0 = \{ 1, 2 \}$, $\Upsilon_1 = \{ 3, 4 \}$,
and $\Upsilon_2 = \{ 5, 6 \}$.
Define the set
\[
\DD = \bigl\{ (m,j) \,:\, m \in [0:2],\;\;
j \in [6] \setminus \Upsilon_m \bigr\}
\]
and the following~$12$ elements of~$F$:
\[
\theta_{m,j} =
\beta_{\kappa(j)} \cdot \prod_{\ell\in\Upsilon_m} \beta_{\kappa(\ell)} ,
\quad (m,j) \in \DD .
\]
We show that out of these~$12$ elements,
there is at least one that differs from all the rest.
Suppose to the contrary that for any given
$(m,j) \in \DD$ there exists
$(m',j') \ne (m,j)$ in~$\DD$ such that
$\theta_{m,j} = \theta_{m',j'}$.
Taking $(m,j) = (0,3)$, we obtain:
\[
\beta_{\kappa(3)} 
\beta_{\kappa(1)} 
\beta_{\kappa(2)} 
= \beta_{\kappa(j')} \prod_{\ell\in\Upsilon_{m'}} \beta_{\kappa(\ell)} .
\]
By the Sidon property we have
$1, 2, 3 \not\in \Upsilon_{m'} \cup \{ j' \}$,
implying that $(m',j') = (2,4)$:
\[
\beta_{\kappa(3)} 
\beta_{\kappa(1)} 
\beta_{\kappa(2)} 
= 
\beta_{\kappa(4)} 
\beta_{\kappa(5)} 
\beta_{\kappa(6)} .
\]
Taking now $(m,j) = (0,4)$ we get, respectively:
\[
\beta_{\kappa(4)} 
\beta_{\kappa(1)} 
\beta_{\kappa(2)} 
= 
\beta_{\kappa(3)} 
\beta_{\kappa(5)} 
\beta_{\kappa(6)} .
\]
The last two equations, in turn, lead to the contradiction:
\[
\left( \beta_{\kappa(4)}/\beta_{\kappa(3)} \right)^2 = 1 .
\]

Our analysis now continues as in Section~\ref{sec:GRSgeneral}.
We assume to the contrary
that under the substitution~(\ref{eq:substitution}),
the elements~(\ref{eq:GRSn-k=3}) sum to zero:
\[
\sum_{(m,j) \in \DD}
\alpha_{\kappa(j)} \cdot
\prod_{\ell \in \Upsilon_m} \alpha_{\kappa(\ell)}
= 0
\]
(compare with~(\ref{eq:sumzero})).
This, in turn, implies
\ifPAGELIMIT
\else
the polynomial identity
\fi
\[
\sum_{(m,j) \in \DD}
\lambda_{\kappa(j)}(z)
\cdot \prod_{\ell \in \Upsilon_m} \lambda_{\kappa(\ell)}(z) = 0
\]
(similarly to~(\ref{eq:basis3})),
and by substituting $x = \xi_i$ we obtain:
\begin{equation}
\label{eq:sumtheta}
\sum_{(m,j) \in \DD} \theta_{m,j}^{2i-1} = 0 , \quad i \in [6]
\end{equation}
(compare with~(\ref{eq:basis4})).
Now, when $\theta_{m,j} = \theta_{m',j'}$
for two index pairs $(m,j) \ne (m',j')$ in~$\DD$,
the respective two terms in~(\ref{eq:sumtheta}) cancel each other;
yet we have shown that at least one of the elements $\theta_{m,j}$
differs from all the rest.
Denoting by $\DD'$ the (nonempty) subset of index pairs in~$\DD$
that remain after such a cancellation, we have:
\[
\sum_{(m,j) \in \DD'} \theta_{m,j}^{2i-1} = 0 , \quad i \in [6] .
\]
However, this is absurd: there can be no $|\DD'| \le 12$
linearly dependent columns in the parity-check matrix
of a binary BCH code whose minimum distance exceeds~$12$.

\ifIEEE
   \appendices
\else
   \section*{$\,$\hfill Appendices\hfill$\,$}
   \appendix
\fi

\section{Skipped proofs}
\label{sec:skipped}

\begin{proof}[Proof of Lemma~\ref{lem:L<q-1}]
We will use the sphere-packing bound along with the inequality
\begin{equation}
\label{eq:V}
V_q(n,\tau) \ge \frac{1}{\sqrt{2n}} \cdot q^{n \Entropy_q(\tau/n)} ,
\end{equation}
where
\[
\Entropy_q(x) =
\left\{
\renewcommand{\arraystretch}{1.5}
\begin{array}{ccl}
\displaystyle
\frac{\entropy(x) + x \, \log_2 (q{-}1)}{\log_2 q}
&& \textrm{if $x \in [0,(q{-}1)/q]$} \\
1
&& \textrm{if $x \in [(q{-}1)/q,1]$}
\end{array}
\right.
\]
(see~\cite[p.~309]{MS}).
The function $x \mapsto \Entropy_q(x)$ is
continuous, non-decreasing, and concave on $[0,1]$,
and is strictly increasing and strictly concave on
$[0,(q{-}1)/q]$, with values ranging from
$\Entropy_q(0) = 0$ to $\Entropy_q((q{-}1)/q) = 1$.

Define the function $x \mapsto f_q(x)$ on $[0,1]$ by
\begin{eqnarray}
\label{eq:fq1}
f_q(x)
& = &
\Entropy_q \left( \frac{q{-}1}{q} \, x \right) - x \\
\label{eq:fq2}
& = &
\frac{1}{\log_2 q}
\left(
\entropy \left( \frac{q{-}1}{q} \, x \right)
- x \cdot \entropy(1/q)
\right) ;
\end{eqnarray}
this function will be of interest since
\begin{equation}
\label{eq:etaq}
\eta_q(\varepsilon)
\stackrel{\textrm{(\ref{eq:eta})}}{=}
(\log_2 q) \cdot f_q(1{-}\varepsilon) .
\end{equation}
We have $f_q(0) = f_q(1) = 0$ which, when combined
with (strict) concavity and continuity, implies that
the minimum of $f_q$
on the interval $[\varepsilon,1{-}\varepsilon]$
is attained at one of the boundaries:
\[
\min_{x \in [\varepsilon,1{-}\varepsilon]} f_q(x) 
= \min \left\{ f_q(\varepsilon), f_q(1{-}\varepsilon) \right\} > 0 .
\]
In fact, for any $\varepsilon \in [0,1/2]$
we always have $f_q(\varepsilon) \ge f_q(1{-}\varepsilon)$,
making $f_q(1{-}\varepsilon)$ the minimum.
To verify this,
we consider the following function 
$\varepsilon \mapsto \varphi_q(\varepsilon)$
on $[0,1/2]$:
\[
\varphi_q(\varepsilon) 
= \entropy \left( \frac{q{-}1}{q} \, (1{-}\varepsilon) \right)
- \entropy \left( \frac{q{-}1}{q} \, \varepsilon \right) .
\]
Differentiating twice with respect to~$\varepsilon$ yields
\[
\frac{d^2 \varphi_q(\varepsilon)}{d \varepsilon^2}
= \frac{q{-}1}{\ln 2}
\cdot
\frac{1 - 2\varepsilon}{\varepsilon (1-\varepsilon)
\left( 1 + (q{-}1) \varepsilon \right)
\left( q - (q{-}1) \varepsilon \right)} ,
\]
which is positive for every $\varepsilon \in (0,1/2)$.
Hence~$\varphi_q$
is convex on $[0,1/2]$ and, since
$\varphi_q(0) = \entropy(1/q)$
and $\varphi_q(1/2) = 0$, we get that it is bounded from above
by the line
$\varepsilon \mapsto (1 - 2 \varepsilon) \cdot \entropy(1/q)$:
\[
\varphi_q(\varepsilon)
\le
(1 - 2 \varepsilon) \cdot \entropy(1/q) .
\]
Therefore,
\[
f_q(1{-}\varepsilon) -
f_q(\varepsilon)
\stackrel{\textrm{(\ref{eq:fq2})}}{=}
\frac{1}{\log_2 q}
\bigl( \varphi_q(\varepsilon) -
(1 - 2 \varepsilon) \cdot \entropy(1/q) \bigr) \le 0 .
\]

Now, let~$\code$ be $(\tau,L)$-list decodable
where~$L$ and~$\tau$
satisfy~(\ref{eq:L<q-1L}) and~(\ref{eq:L<q-1tau}),
and suppose in addition that $L \ge q-1$.
By Theorem~\ref{thm:SpherePacking}, the inequality~(\ref{eq:V}),
and the monotonicity of $\Entropy_q(\cdot)$ we must have
\begin{eqnarray}
1 - R
& \ge &
\Entropy_q \left( \frac{L}{L+1} (1 - R \right)
- \frac{\log_q \left( L \sqrt{2n} \right)}{n} \nonumber \\
\label{eq:1-R}
& \!\!\!\!\stackrel{L \ge q-1}{\ge}\!\!\!\! &
\Entropy_q \left( \frac{q-1}{q} (1 - R \right)
- \frac{\log_q \left( L \sqrt{2n} \right)}{n} .
\quad
\end{eqnarray}
Thus, for $R \in [\varepsilon,1{-}\varepsilon]$,
\begin{eqnarray*}
f_q(1{-}\varepsilon)
& = &
\min_{x \in [\varepsilon,1{-}\varepsilon]} f_q(x) 
\; \le \; f_q(1-R) \\
& \!\stackrel{\textrm{(\ref{eq:fq1})+(\ref{eq:1-R})}}{\le}\! &
\frac{\log_q \left( L \sqrt{2n} \right)}{n}
\;\; \stackrel{\textrm{(\ref{eq:L<q-1L})}}{<} \;\;
\frac{\eta_q(\varepsilon)}{\log_2 q} ,
\end{eqnarray*}
thereby contradicting~(\ref{eq:etaq}).
Hence, we must have $L < q-1$.

Finally, we show that $\eta_q(\varepsilon)$ is bounded from below
by $\eta_2(\varepsilon)$ (which is positive and independent of~$q$).
We consider the bivariate function
$(\varepsilon,q) \mapsto \eta_q(\varepsilon)$
as if~$q$ were real and differentiate with respect to~$q$ to obtain
\[
\frac{\partial \eta_q(\varepsilon)}{\partial q}
= \frac{1{-}\varepsilon}{q^2}
\cdot \log_2 \left( \frac{q}{1{-}\varepsilon} - (q{-}1) \right) .
\]
This derivative is positive for any $\varepsilon \in (0,1)$ and, so,
$\eta_q(\varepsilon) \ge \eta_2(\varepsilon)$ for any field size~$q$.
\end{proof}

\begin{proof}[Proof of Lemma~\ref{lem:nesting}]
The lemma is immediate when $\ell \le t$, so we assume
hereafter in the proof that $\ell < \binom{w}{t} / \binom{w-s}{t}$.
We construct the set~$X$ using the following iterative procedure.
We start with $X \leftarrow \emptyset$,
$Y \leftarrow [w]$,
and $\LL \leftarrow \{ J_m \}_{m \in [\ell]}$,
and for $i \leftarrow 1, 2, \ldots, t$,
we select to $X$ an element $y \in Y$ for which the size of
\begin{equation}
\label{eq:containing}
\bigl\{ J_m \in \LL \,:\, y \in J_m \bigr\}
\end{equation}
is the largest (namely, there is no element in~$Y$ that is contained
in more subsets $J_m \in \LL$ than~$y$).
We then remove from $\LL$ the subsets $J_m$ that contain~$y$
and remove~$y$ from~$Y$.
We will show that after no more than~$t$ iterations,
the set~$\LL$ becomes empty.

For $i \in [0:t]$, let $\ell_i$ denote the size of $\LL$
right after the $i$th iteration, with $\ell_0 = \ell$.
Right after that iteration we have
\[
\sum_{J_m \in \LL} |J_m| \ge \ell_i \cdot s ,
\]
which means that during the $(i+1)$st iteration,
the size of~(\ref{eq:containing}) for a maximizing~$y$
is bounded from below by:
\[
\Bigl\lceil \frac{1}{w-i} \sum_{J_m \in \LL} |J_m| \Bigr\rceil
\ge \left\lceil \frac{\ell_i \cdot s}{w-i} \right\rceil .
\]
This leads to the inequality
\[
\ell_{i+1} \le \ell_i
- \left\lceil \frac{\ell_i \cdot s}{w-i} \right\rceil
= \left\lfloor \ell_i \cdot \frac{w-i-s}{w-i} \right\rfloor
\]
and, ignoring the integer truncation, we get the upper bound:
\[
\ell_t \le \ell_0 \cdot \prod_{i=0}^{t-1} \frac{w-i-s}{w-i}
= \ell \cdot \binom{w-s}{t} \biggm/ \binom{w}{t} .
\]
But $\ell < \binom{w}{t} / \binom{w-s}{t}$ and, so $\ell_t < 1$,
which means that $\ell_t = 0$.
\end{proof}

\begin{proof}[Proof of Lemma~\ref{lem:vartheta}]
The case $k = 1$ is easily verified, so
we assume hereafter that $k > 1$.

Let $w \in [n{-}k{-}2]$; then
\[
n-w \ge k+2 \ge k + 1 + \frac{1}{k-1} = \frac{k^2}{k-1}
\]
and, so,
\[
n-w \le k(n{-}k{-}w) .
\]
The latter inequality, in turn, is equivalent to
\[
(n{-}k{-}w)^2 \le (n{-}k{-}w{-}1)(n-w)
\]
which, in turn, is equivalent to
\[
(n{-}k{-}w) \binom{n-w-1}{k} \le (n{-}k{-}w{-}1) \binom{n-w}{k} .
\]
Making the inequality strict by adding~$1$
to the right-hand side yields
\[
(n{-}k{-}w) \binom{n-w-1}{k} < (n{-}k{-}w{-}1) \binom{n-w}{k} + 1 ,
\]
which can be equivalently written as
\begin{eqnarray*}
\lefteqn{
(n{-}k{-}w) 
\left( \binom{n-w-1}{k} - (n{-}k{-}w) \right)
} \makebox[5ex]{} \\
& < &
(n{-}k{-}w{-}1) \left( \binom{n-w}{k} - (n{-}k{-}w{+}1) \right) .
\end{eqnarray*}
Finally, dividing both sides by $(n{-}k{-}w)(n{-}k{-}w{+}1)$
yields
\[
\vartheta_{w+1} < \vartheta_w .
\]
\end{proof}

\begin{proof}[Proof of Lemma~\ref{lem:L=q=2}]
The ``if'' part follows from
Examples~\ref{ex:LMDSn-k=2} and~\ref{ex:LMDSk=1}.
Turning to the ``only if'' part, let~$\code$ be
a linear $[n,k]$ code over~$F$ where $2 \le k \le n-2$
and suppose, without loss of generality, that~$\code$
has a $k \times n$ generator matrix of the form
$G = \left( I \,|\, A \right)$.
We show that~$\code$ cannot be $2$-MDS by distinguishing
between several cases, according to the number of zero entries
in~$A$. Denote by $\bldg_i$ the $i$th row of~$G$.

\emph{Case 1: $A$ contains at least two zero entries
(say, within the first two rows).}
Take
$\blde_0 = \bldzero$,
$\blde_1 = \bldg_1$, and
$\blde_2 = \bldg_2$.
These three vectors are in the same (trivial) coset~$\code$ and
$\sum_{m \in [0:2]} \weight(\blde_m) \le 2(n-k)$.

\emph{Case 2: $A$ contains one zero entry
(say, the last entry in the first row).}
Take
$\blde_0 = \left( \bldzero_k \,|\, \bldone_{n-k-1} \,|\, 0 \right)$
(with~$k$ leading zeros),
$\blde_1 = \blde_0 + \bldg_1$,
and
$\blde_2 = \blde_0 + \bldg_2$.
These three vectors are in the same coset of~$\code$ and
$\sum_{m \in [0:2]} \weight(\blde_m) = n-k+2 \le 2(n-k)$.

\emph{Case 3: $A$ contains no zero entries (i.e., it is all-$1$s).}
Take
$\blde_0 = \left( \bldzero_k \,|\, \bldone_{n-k} \right)$,
$\blde_1 = \blde_0 + \bldg_1$,
and
$\blde_2 = \blde_0 + \bldg_2$.
Again, these three vectors are in the same coset of~$\code$ and
$\sum_{m \in [0:2]} \weight(\blde_m) = n-k+2 \le 2(n-k)$.
\end{proof}

\begin{proof}[Proof of Theorem~\ref{thm:almostall2-MDS}]
We first recall that for
a uniformly distributed random $(n-k) \times n$ matrix
$\randomH$ over~$F$ we have\footnote{%
We use boldface letters for random matrices and normal font
for their realizations.}
(see~\cite[pp.~444--445]{MS}):
\begin{eqnarray}
\lefteqn{
\Prob \bigl\{ \rank(\randomH) < n-k \bigr\}
= 1 - \prod_{i=0}^{n-k-1} \left( 1 - q^{i-n} \right)
} \makebox[15ex]{} \nonumber \\
\label{eq:singular}
& \le &
\sum_{i=0}^{n-k-1} q^{i-n}
< \frac{1}{q^k(q-1)} .
\end{eqnarray}

In what follows, we consider uniformly distributed
random linear $[n,k]$ codes over~$F$; yet, we will find
it more convenient to replace this probability space by
that of uniformly distributed random $(n-k) \times n$ matrices over~$F$,
which will stand for the parity-check matrices of the codes.
These two probability spaces would be the same if we conditioned
the matrices to be of full rank $n-k$; however, since we seek
only an upper bound on the probability of a code being non-$2$-MDS,
the effect of this conditioning amounts to a constant factor
(which approaches~$1$ as~$q$ increases). Specifically,
for any $(n-k) \times n$ matrix~$H$ over~$F$:
\begin{eqnarray*}
\lefteqn{
\Prob \bigl\{ \randomH = H  \,|\, \rank(\randomH) = n-k \bigr\}
} \makebox[10ex]{} \\
& \le &
\frac{\Prob \left\{ \randomH = H \right \}}%
              {\Prob \left\{ \rank(\randomH) = n-k \right\}} \\
& \stackrel{\mathrm{(\ref{eq:singular})}}{<} &
\frac{1}{1 - q^{-k}/(q-1)} \cdot \Prob \left\{ \randomH = H \right\} .
\end{eqnarray*}
For a random $(n-k) \times n$ matrix $\randomH$,
we let $\code(\randomH) = \ker(\randomH)$ be the code over~$F$
which is defined by the parity-check matrix~$\randomH$.

Denote by $\Event_\MDS$ the event
that every set of $n-k$ columns in a random $(n-k) \times n$
matrix~$\randomH$ is linearly independent
(namely, that~$\randomH$ is a parity-check matrix of an MDS code).
A union bound yields the following upper bound on
the probability of the complement event:
\begin{equation}
\label{eq:notMDS}
\Prob \bigl\{ \overline{\Event_\MDS} \bigr\}
\stackrel{\mathrm{(\ref{eq:singular})}}{<}
\binom{n}{n-k} \cdot \frac{q^{-k}}{q-1} \le \frac{2^n}{q} .
\end{equation}
Thus, by Examples~\ref{ex:LMDSn-k=2}--\ref{ex:LMDSk=2,L=2},
the theorem holds when $\min \{ k, n-k \} \le 2$.
We assume hereafter in the proof that $3 \le k \le n-3$.

Fix a subset $J \subseteq [n]$ of size $|J| = w$
in the (nonempty) range~(\ref{eq:wrange})
and write $J' = [n] \setminus J$ and $r = n-k-w$.
Also, fix three nonzero row vectors
$\blde_0, \blde_1, \blde_2 \in F^{n-w}$, with entries indexed
by~$J'$, such that their supports are disjoint,
$\sum_{m \in [0:2]} \weight(\blde_m) \le 2r$,
and the first nonzero entry in $\blde_0$ equals~$1$.

Denote by $\Event_J$ the event that
$\rank \left( (\randomH)_J \right) = w$;
note that $\Event_\MDS \subseteq \Event_J$.
Also, let $\randomH^* = (\randomP \randomH)_{J'}$,
where~$\randomP$ is an $r \times (n-k)$ matrix
whose rows form a basis of the left kernel\footnote{%
The matrix $\randomP$ will be used only under conditioning
on $\Event_J$, in which case it indeed has $r$ rows.}
of $(\randomH)_J$;
we recall from the proof of Theorem~\ref{thm:puncturingL=2}
that $\randomH^*$ is the parity-check matrix
of the code $(\code(\randomH))_{J'}$,
which is obtained by puncturing $\code(\randomH)$
on the coordinate set~$J$.

Define
\begin{eqnarray*}
\lefteqn{
\PP(\blde_0,\blde_1,\blde_2;J)
} \makebox[2ex]{} \\
& = &
\Prob \bigl\{
\left(\randomH^* \blde_0^\top
= \randomH^* \blde_1^\top
= \randomH^* \blde_2^\top\right) \cap \Event_\MDS
\bigr\} ,
\end{eqnarray*}
which is the probability that
$\code(\randomH)$ is MDS
yet the three vectors $\blde_0$, $\blde_1$, and $\blde_2$
are in the same coset of the punctured code
$(\code(\randomH))_{J'}$.
We have:
\begin{eqnarray}
\lefteqn{
\PP(\blde_0,\blde_1,\blde_2;J)
} \makebox[0ex]{} \nonumber \\
& \le &
\Prob \left\{
\randomH^* \blde_0^\top
= \randomH^* \blde_1^\top
= \randomH^* \blde_2^\top \,|\, \Event_J \right\}
\cdot \Prob \left\{ \Event_J \right\} \nonumber \\
\label{eq:samecosetagain}
& = &
q^{-2r} \cdot \Prob \left\{ \Event_J \right\}
\le q^{-2r} ,
\end{eqnarray}
where the equality follows from the following two facts:
\begin{list}{}{\settowidth{\labelwidth}{\textrm{(a)}}}
\item[(a)]
the random $r \times (n-w)$ matrix~$\randomH^*$ is
uniformly distributed (even under the conditioning on $\Event_J$,
since none of the elements in~$J$ indexes
any column of $\randomH^*$), and---
\item[(b)]
the vectors $\blde_m$ have disjoint supports, and, so,
the respective syndromes $\randomH^* \blde_m^\top$
are statistically independent.
\end{list}

Let now $\Event_\light$ be the event that every puncturing
of $\code(\randomH)$ on any~$w$ coordinates
in the range~(\ref{eq:wrange}) results in a lightly-$2$-MDS code.
By a union bound we get:
\begin{eqnarray*}
\lefteqn{
\Prob \left\{ 
\overline{\Event_\light} \cap \Event_\MDS
\right\}
} \makebox[0ex]{} \\
& \le &
\sum_J
\sum_{(\blde_0,\blde_1,\blde_2)}
\PP(\blde_0,\blde_1,\blde_2;J) \\
& = &
\sum_{J}
\sum_{(J_0,J_1,J_2)}
\sum_{\genfrac{}{}{0pt}{}{(\blde_0,\blde_1,\blde_2)\,:}%
                                              {\Support(\blde_m) = J_m}}
\PP(\blde_0,\blde_1,\blde_2;J) ,
\end{eqnarray*}
where~$J$ ranges over all subsets of $[n]$
of size~$w$ in the range~(\ref{eq:wrange})
and $(J_0,J_1,J_2)$ ranges over all triples of nonempty disjoint
subsets of $J'$ such that $\sum_{m \in [0:2]} |J_m| \le 2r$.
Hence,
\begin{eqnarray}
\lefteqn{
\Prob \left\{ \overline{\Event_\light} \cap \Event_\MDS \right\}
} \makebox[5ex]{} \nonumber \\
& \stackrel{\mathrm{(\ref{eq:samecosetagain})}}{\le} &
\frac{q^{-2r}}{q-1}
\sum_J \sum_{(J_0,J_1,J_2)}
\prod_{m \in [0:2]} (q-1)^{|J_m|} \nonumber \\
& \le &
\frac{1}{q-1} \left( \frac{q-1}{q} \right)^{2r}
\sum_J \sum_{(J_0,J_1,J_2)} 1
\nonumber \\
\label{eq:Oterm}
& \le &
\frac{5^n}{q-1} \left( \frac{q-1}{q} \right)^{2r}
= O \left( 5^n/q \right) .
\end{eqnarray}
We conclude that
\begin{eqnarray*}
\lefteqn{
\Prob \left\{ \overline{\Event_\light \cap \Event_\MDS} \right\}
} \makebox[0ex]{} \\
& = &
\Prob \left\{ \overline{\Event_\light} \cup \overline{\Event_\MDS}
\right\} \\
& = &
\Prob \left\{ \overline{\Event_\light} \cap \Event_\MDS \right\}
+
\Prob \left\{ \overline{\Event_\MDS} \right\} \\
& \stackrel{\mathrm{(\ref{eq:notMDS})+(\ref{eq:Oterm})}}{=} &
O \left( 5^n/q \right) ,
\end{eqnarray*}
and the sought result follows from Theorem~\ref{thm:puncturingL=2}.
\end{proof}

\begin{remark}
\label{rem:almostall2-MDS}
The upper bound, $O(5^n/q)$, on the fraction of non-$2$-MDS codes
in Theorem~\ref{thm:almostall2-MDS} is not the best possible;
e.g., it can be improved when $k \ne n/2$, yet we omit
the details.\qed
\end{remark}

\section{Properties of the determinant of $M_\bldrho(\bldx)$}
\label{sec:M}

Given an integer $\rho \ge 3$,
let $\bldrho = (\rho_0,\rho_1,\rho_2)$ be a partition of $2\rho$
where $2 \le \rho_0 \le \rho_1 \le \rho_2 < \rho$
and let the subsets~$\Upsilon_0$, $\Upsilon_1$, and~$\Upsilon_2$
be defined as in~(\ref{eq:Upsilon}).
Let $\bldx = (x_\ell)_{\ell \in [2\rho]}$
be a vector of indeterminates
and define the matrix $M_\bldrho(\bldx)$
as in~(\ref{eq:M(x)}).
For each $m \in [0:2]$, write $\bldx_m = (\bldx)_{\Upsilon_m}$
and define the following
degree-$\rho_m$ polynomial in~$z$ over $F[\bldx_m]$:
\[
\sigma_m(z) = \sigma_m(z;\bldx_m)
= \prod_{\ell \in \Upsilon_m} (z - x_\ell)
= \sum_{j=0}^{\rho_m} \sigma_{m,j} z^{\rho_m-j} ,
\]
where $\sigma_{m,j} = \sigma_{m,j}(\bldx_m)$
is a degree-$j$ homogeneous polynomial in $F[\bldx_m]$
(in particular, $\sigma_{m,0} = 1$ and
$\sigma_{m,\rho_m} = \prod_{\ell \in \Upsilon_m} (-x_\ell)$).
We associate with $\sigma_m(z)$
the following $(\rho - \rho_m) \times \rho$ matrix over $F[\bldx_m]$,
\[
\newcommand{\sdots}{\raisebox{-0.5ex}{$\vdots$}}
E_\rho(\bldx_m) =
\left(
\arraycolsep0.48ex
\begin{array}{ccccccc}
\sigma_{m,0} & \sigma_{m,1} & \cdots & \sigma_{m,\rho_m} &0&\cdots&0 \\
0 & \sigma_{m,0} & \sigma_{m,1} & \cdots & \sigma_{m,\rho_m} &0&\vdots\\
\vdots  & \ddots        & \ddots       & \ddots & \cdots & \ddots & 0 \\
0&\cdots&0& \sigma_{m,0} & \sigma_{m,1} & \cdots & \sigma_{m,\rho_m}
\end{array}
\right) ,
\]
and define the $\rho \times \rho$ matrix
$S_\bldrho(\bldx)$ by stacking the matrices
$E_\rho(\bldx_m)$, $m \in [0:2]$:
\begin{equation}
\label{eq:stack}
S_\bldrho(\bldx) =
\renewcommand{\arraystretch}{1.3}
\left(
\begin{array}{c}
E_\rho(\bldx_0) \\ \hline
E_\rho(\bldx_1) \\ \hline
E_\rho(\bldx_2) \end{array}
\right)
\end{equation}
(note that the number of rows in $S_\bldrho(\bldx)$ is indeed
$\sum_{m \in [0:2]} (\rho - \rho_m) = \rho$).
This matrix can be seen as the generalization to three polynomials---%
namely, $\sigma_0(z)$, $\sigma_1(z)$, and $\sigma_2(z)$---%
of the Sylvester matrix of two polynomials;
the determinant of the Sylvester matrix equals the resultant
of the two polynomials~\cite[\S 9.2]{Zippel}.

We have the following theorem,
which allows us to use the matrix $S_\bldrho(\bldx)$
when applying Lemma~\ref{lem:lightly}
to test if a given GRS code is lightly-$2$-MDS.

\begin{theorem}
\label{thm:Sylvester}
Using the notation above,
the following two conditions are equivalent for every vector
$\bldalpha = (\alpha_\ell)_{\ell \in [2\rho]}$
over any extension field of~$F$.
\begin{list}{}{\settowidth{\labelwidth}{\textit{(ii)}}}
\itemsep.5ex
\item[(i)]
$\det(M_\bldrho(\bldalpha)) \ne 0$.
\item[(ii)]
$\det(S_\bldrho(\bldalpha)) \ne 0$,
and for every $m \in [0:2]$, the entries of
$(\bldalpha)_{\Upsilon_m}$ are all distinct.
\end{list}
\end{theorem}

\begin{proof}
Let~$\Phi$ be an extension field of~$F$ which contains
the entries of~$\bldalpha$ and,
for $m \in [0:2]$, write $\bldalpha_m = (\bldalpha)_{\Upsilon_m}$.
Clearly, if there are two identical entries in $\bldalpha_m$
for some $m \in [0:2]$ then $\det(M_\bldrho(\bldalpha)) = 0$.
Thus, we assume hereafter in the proof that for each $m \in [0:2]$,
all the entries of $\bldalpha_m$ are distinct.
We denote by $\Phi_\rho[z]$ the set of polynomials
(in the indeterminate~$z$) of degree less than~$\rho$ over~$\Phi$.

The condition $\det(M_\bldrho(\bldalpha)) = 0$
is equivalent to having two row vectors,
$\blda_1, \blda_2 \in \Phi^\rho$, not both zero, such that
\[
\left(\blda_1 \,|\, \blda_2 \right)
M_\bldrho(\bldalpha) = \bldzero .
\]
This condition, in turn, is equivalent to having two polynomials,
$a_1(z), a_2(z) \in \Phi_\rho[z]$, not both zero, such that
\begin{equation}
\label{eq:Sylvester1}
a_0(\alpha_\ell) + a_1(\alpha_\ell) = 0 , \quad \ell \in \Upsilon_0
\end{equation}
and, for $m = 1, 2$:
\begin{equation}
\label{eq:Sylvester2}
a_m(\alpha_\ell) = 0 , \quad \ell \in \Upsilon_m
\end{equation}
(namely, $\blda_m$ is the vector of coefficients of $a_m(z)$, with
the first entry in $\blda_m$ being the free coefficient of $a_m(z)$).
Conditions~(\ref{eq:Sylvester1}) and~(\ref{eq:Sylvester2})
are equivalent to having three polynomials,
$a_0(z), a_1(z), a_2(z) \in \Phi_\rho[z]$, not all zero, such that
\[
\sum_{m \in [0:2]} a_m(z) = 0
\]
and, for $m \in [0:2]$:
\[
a_m(\alpha_\ell) = 0 , \quad \ell \in \Upsilon_m .
\]
The latter equation means
that $\sigma_m(z;\bldalpha_m)$ is a divisor of $a_m(z)$
in $\Phi[z]$, namely, for each $m \in [0:2]$ there exists
$b_m(z) \in \Phi_{\rho-\rho_m}(z)$ such that
\[
a_m(z) = b_m(z) \cdot \sigma_m(z;\bldalpha_m) .
\]
We conclude that the condition $\det(M_\bldrho(\bldalpha)) = 0$
is equivalent to having polynomials
$b_m(z) \in F_{\rho-\rho_m}(z)$, $m \in [0:2]$, not all zero, such that
\[
\sum_{m \in [0:2]} b_m(z) \cdot \sigma_m(z;\bldalpha_m) = 0 .
\]
Denoting by $\bldb_m$ the vector of coefficients
(in $\Phi^{\rho-\rho_m}$) of $b_m(z)$,
with the free coefficient being the \emph{last} entry in $\bldb_m$,
the last equation can be written in vector form as
\[
\sum_{m \in [0:2]} \bldb_m \cdot E_\rho(\bldalpha_m) = \bldzero ,
\]
or as
\[
\left(\bldb_0 \,|\, \bldb_1 \,|\, \bldb_2 \right)
S_\bldrho(\bldalpha) = \bldzero ,
\]
where $\bldb_0, \bldb_1, \bldb_2$ are not all zero.
Yet this can hold if and only if $\det(S_\bldrho(\bldalpha)) = 0$.
\end{proof}

It follows from Theorem~\ref{thm:Sylvester} that
$\bldalpha \; (\in \Phi^{2 \rho})$ is
a root of $\det(M_\bldrho(\bldx)) \; (\in F[\bldx])$
if and only if it is a root of
\[
\Gamma_\bldrho(\bldx) = \det(S_\bldrho(\bldx)) \cdot
\prod_{m \in [0:2]}
\prod_{\genfrac{}{}{0pt}{}{\ell,\ell' \in \Upsilon_m \,:}{\ell > \ell'}}
(x_\ell - x_{\ell'}) .
\]
In particular, $\det(M_\bldrho(\bldx))$, as an element of $F[\bldx]$,
is divisible by $x_\ell - x_{\ell'}$
for every two indexes $\ell \ne \ell'$ in the same subset~$\Upsilon_m$.
Based on numerical results, we conjecture that
\begin{equation}
\label{eq:conjecture}
\det(M_\bldrho(\bldx)) 
= (-1)^{\rho (\rho_1 + 1)} \cdot \Gamma_\bldrho(\bldx) .
\end{equation}

One evidence that supports this conjecture is that
both $\det(M_\bldrho(\bldx))$ and $\Gamma_\bldrho(\bldx)$
have the same total degree.
Specifically, we pointed out in Section~\ref{sec:GRSnotation}
that the Leibniz expansion of $\det(M_\bldrho(\bldx))$
results in a sum of monomials~(\ref{eq:monomial})
that satisfy conditions~(R1) and~(R2) and, therefore, each
has total degree $\rho(\rho-1)$.
As for the degree of $\Gamma_\bldrho(\bldx)$,
we have the following result.

\begin{lemma}
\label{lem:degSylvester}
$\det(S_\bldrho(\bldx))$
is a homogeneous polynomial in $F[\bldx]$ of total degree
\[
\rho(\rho-1) - \sum_{m \in [0:2]} \binom{\rho_m}{2}
= \rho^2 - \frac{1}{2} \sum_{m \in [0:2]} \rho_m^2 .
\]
\end{lemma}

\begin{proof}
A typical term in the Leibniz expansion of
$\det(S_\bldrho(\bldx))$
takes---up to a sign---the form
\begin{equation}
\label{eq:typicalterm}
\prod_{m \in [0:2]}
\prod_{i \in [\rho-\rho_m]}
\sigma_{m,j_m(i)-i} ,
\end{equation}
where the list
\[
\left(
(j_0(i))_{i \in [\rho-\rho_0]}
\,\Bigm|\,
(j_1(i))_{i \in [\rho-\rho_1]}
\,\Bigm|\,
(j_2(i))_{i \in [\rho-\rho_2]}
\right)
\]
forms a permutation on $[\rho]$.
The term~(\ref{eq:typicalterm}),
when expressed as a multivariate polynomial in
the entries of~$\bldx$, has total degree
\begin{eqnarray*}
\lefteqn{
\sum_{m \in [0:2]}
\sum_{i \in [\rho-\rho_m]}
(j_m(i) - i)
} \makebox[0ex]{} \\
& = &
\Bigl( \sum_{m \in [0:2]}
\sum_{i \in [\rho-\rho_m]}
j_m(i) \Bigr)
-
\Bigl( \sum_{m \in [0:2]}
\sum_{i \in [\rho-\rho_m]}
i \Bigr) \\
& = &
\Bigl( \sum_{j \in [\rho]} j \Bigr)
-
\Bigl( \sum_{m \in [0:2]}
\sum_{i \in [\rho-\rho_m]}
i \Bigr) \\
& = &
\binom{\rho+1}{2} - \sum_{m \in [0:2]} \binom{\rho-\rho_m+1}{2} \\
& = &
\rho^2 - \frac{1}{2} \sum_{m \in [0:2]} \rho_m^2 .
\end{eqnarray*}
\end{proof}

It follows from the lemma that the equality~(\ref{eq:conjecture})
holds (up to a multiplying scalar) if $\det(S_\bldrho(\bldx))$ is
absolutely irreducible (namely, it is irreducible over
any extension field of~$F$)~\cite[p.~4]{Shafarevich};
however, this is yet to be shown (or to be disproved).\footnote{%
In the case of Example~\ref{ex:GRSn-k=3} it is
absolutely irreducible.
Note that this is in contrast with the case of the resultant of
two polynomials, which is highly reducible:
it factors into linear terms over
the splitting field of the polynomials~\cite[p.~142]{Zippel}.
Similar behavior of determinants
is seen also in other problems (see~\cite{AS} or~\cite{AT}).}

The definitions of $M_\bldrho(\bldx)$
and $S_\bldrho(\bldx)$ can be generalized
in a straightforward manner to any $L \ge 2$.
Given a partition
$\bldrho = (\rho_0,\rho_1,\ldots,\rho_L)$ of $L \rho$ such that
$L \le \rho_0 \le \rho_1 \le \cdots \le \rho_L < \rho$,
the matrix $M_\bldrho(\bldx)$ will then have order
$(L \rho) \times (L \rho)$
and the matrix $S_\bldrho(\bldx)$,
which is obtained by stacking $L+1$ matrices
(instead of three) in~(\ref{eq:stack}),
will have order $\rho \times \rho$.
Theorem~\ref{thm:Sylvester} generalizes accordingly
and, therefore, by Lemma~\ref{lem:lightly},
we can use the matrix $S_\bldrho(\bldx)$
to test if a given GRS code is lightly-$L$-MDS.
The total degree of $\det(S_\bldrho(\bldx))$ is
\[
L \cdot \binom{\rho}{2} - \sum_{m \in [0:L]} \binom{\rho_m}{2}
= \frac{L}{2} \cdot \rho^2 - \frac{1}{2} \sum_{m \in [0:L]} \rho_m^2
\]
and, based on numerical evidence, we conjecture that
\[
\det(M_\bldrho(\bldx)) 
=
(-1)^s \cdot
\det(S_\bldrho(\bldx)) 
\cdot
\prod_{m \in [0:L]}
\prod_{\genfrac{}{}{0pt}{}{\ell,\ell' \in \Upsilon_m \,:}{\ell > \ell'}}
(x_\ell - x_{\ell'}) ,
\]
where
\[
s = \rho \cdot
\left( \sum_{\mathrm{odd} \; m \in [0:L]} (\rho - \rho_m) \right) .
\]

For a certain structure of assignments to the vector~$\bldx$,
the matrix $S_\bldrho(\bldx)$ also appears in~\cite[\S\S III--IV]{RV},
as part of the analysis of the burst-error
list decodability of (cyclic) Reed--Solomon codes.

\section*{Acknowledgment}

I would like to thank the authors of~\cite{GST}
for making me aware of their work.
\fi

\end{document}